\documentclass[runningheads]{llncs}
\usepackage{graphicx}
\usepackage{natbib}
\usepackage{subcaption}
\usepackage[margin=1in]{geometry}
\usepackage{amssymb,thmtools,thm-restate,mathtools}
\usepackage{pgf,tikz,pgfplots}
\pgfplotsset{compat=1.15}
\usetikzlibrary{arrows}
\usepackage{pgfplots}
\usetikzlibrary{patterns}
\usepgfplotslibrary{fillbetween}
\usetikzlibrary{intersections}
\renewcommand{\refname}{References}

\usepackage[colorlinks=true,linkcolor=olive, citecolor=blue]{hyperref}
\usepackage{cleveref}
\newcommand{\inventory}{k}
\newcommand{\customer}{n}
\newcommand{\EX}{\mathbb{E}}
\newcommand{\prob}{\mathbb{P}}
\DeclareMathOperator*{\argmax}{arg\,max}

\crefname{rmk}{remark}{remarks}
\Crefname{rmk}{Remark}{Remarks}

\usepackage{graphicx} % Required for inserting images

\begin{document}

\title{Dynamic Run-Out Pricing for Perishable Goods}

\author{Farbod Ekbatani}
\institute{University of Chicago Booth School of Business}
\authorrunning{F. Ekbatani}

\maketitle

\begin{abstract}
We study a dynamic pricing problem in which a seller with limited inventory must sequentially price a perishable product for arriving customers with heterogeneous valuations. The key difficulty lies in balancing immediate sales against future opportunities under capacity constraints and stochastic customer valuations. We propose a simple, interpretable, and asymptotically optimal dynamic pricing algorithm, called Dynamic Run-Out Pricing (DROP), that adjusts prices in real time based on current inventory and anticipated future demand. 

Under the assumption that customers' valuations follow bounded or exponential distributions, our theoretical results show that DROP's expected revenue converges to the optimal revenue as market size grows. Unlike traditional notions of near-optimality, where the ratio of the algorithm's revenue to the optimal revenue converges to 1, we show that the absolute difference between the two revenues converges to 0 as the market size increases, underscoring the algorithm's effectiveness.

Furthermore, we extend our analysis to heavy-tailed valuation distributions, specifically the Pareto distribution, where standard run-out heuristics severely underprice inventory by ignoring the value of waiting for extreme valuations. To correct this, we introduce the Tail-Aware DROP (TA-DROP) heuristic, which intentionally inflates prices to protect inventory, and we prove that it restores asymptotic optimality with a strict polynomial convergence rate. 

We complement the theoretical results with numerical experiments to demonstrate the practical applicability and competitive performance of these algorithms across a range of dynamic pricing scenarios.

\end{abstract}

% \end{frontmatter}
\section{Introduction}
Across numerous industries, a fundamental operational challenge is the effective allocation of limited capacity to customers with heterogeneous valuations and uncertain arrival times. Dynamic pricing addresses this challenge by allowing sellers to adjust prices in real time, seamlessly aligning fixed supply with stochastic demand to maximize revenue. Yet, the design of optimal pricing policies is inherently complex, as accurately capturing the underlying stochastic dynamics often drives sellers toward highly sophisticated models.

Existing pricing strategies typically fall at opposite ends of a complexity spectrum. Simple, rule-based models are easy to deploy, yet often result in significant revenue loss due to their lack of sophistication. In contrast, advanced methods, such as those based on deep learning, can approach optimal pricing performance but are frequently complex, difficult to interpret, and impractical for routine use. This creates a fundamental trade-off between interpretability and optimality: sellers must often choose between sacrificing potential profits and maintaining operational simplicity.

To navigate this trade-off, this paper proposes a dynamic pricing strategy that achieves both near-optimality and strict interpretability. Our approach adjusts prices in real time based on current inventory levels and anticipated customer arrivals, empowering sellers to maximize revenue without sacrificing operational simplicity.

Our work is motivated by a collaboration with Oracle Nor1.\footnote{For more details, see \href{https://docs.oracle.com/en/industries/hospitality/nor1-cloud/}{NOR1}} Oracle Nor1’s solution enables hotels to offer discounted room upgrades at check-in, enhancing guest experience and generating incremental revenue. A unique characteristic of this context is that the seller possesses an accurate estimate of the total pool of potential upgrade customers: only guests who have already booked standard rooms are eligible. This certainty, rare in dynamic pricing settings, allows for more targeted and efficient pricing strategies.

We focus on a single-item revenue management setting, where a small inventory of premium rooms can be allocated among a large pool of potential upgrade customers. This scenario, which is typical in hotel upgrades---for example, a hotel with 200 standard room guests and only a handful of available premium rooms---intensifies the need for careful resource management. Although this framework abstracts away the additional complexity of multi-night bookings, it provides a tractable and practically relevant foundation for developing and analyzing our pricing algorithm.

To ensure analytical tractability and derive closed-form insights into the algorithm's performance, we analyze our pricing framework under three distinct classes of willingness-to-pay distributions: bounded, exponential, and Pareto. By comparing these regimes, we illustrate how the tail behavior of customer demand fundamentally alters the optimal pricing dynamics. 

We demonstrate that the standard DROP algorithm achieves absolute asymptotic optimality for both bounded and light-tailed (exponential) distributions. However, we show that it falls short when faced with the heavy-tailed Pareto distribution, where standard run-out heuristics ignore the immense option value of extreme valuations and underprice inventory. To bridge this gap, we introduce the Tail-Aware DROP (TA-DROP) heuristic, proving that an intentional, theoretically grounded inflation of prices restores asymptotic optimality with a strict polynomial convergence rate.

\subsection{Literature Review}

The revenue management of perishable goods has been the subject of extensive research, with early work emphasizing the efficacy of simple pricing strategies. \citet{Gvr} introduced a fixed-price strategy that maximizes revenue without requiring continuous adjustments throughout the sales period. Notably, their results show that such static pricing can perform near-optimally, particularly when both the demand and inventory are high, making it highly relevant in practical environments where flexible pricing is challenging. Similarly, within the computer science literature, \citet{Qiqi} employed the concept of the correlation gap to quantify the performance difference between static and dynamically optimized pricing. Specifically, under the assumption of independent customer valuations, they demonstrated a fixed inefficiency gap of $1-\frac{1}{e}$ in the single-item case between the static-fluid benchmark and the true optimum, underscoring the inherent limitations of static pricing in stochastic settings.

Subsequent studies have explored dynamic pricing, where prices are adapted based on real-time market conditions. \citet{MM} demonstrated the value of re-optimizing prices, especially in high-traffic circumstances. \citet{CF} extended this exploration to environments with forecast uncertainty, illustrating how dynamic re-optimization enhances resilience to market shocks and demand variations. \citet{Jasin} provided analytical support that time-varying (piecewise-static) prices can significantly outperform purely static approaches. Building on this approach of repeated re-optimization, \citet{wang2022constant} analyze price-based network revenue management, demonstrating that frequently re-solving the fluid model yields strong theoretical performance against the exact optimal policy. However, their analysis is fundamentally rooted in the fluid approximation over a finite time horizon $T$. In contrast, our work operates in the large market regime. As our framework benefits from a fixed pool of potential buyers $N$, this demand predictability allows us to derive explicit, closed-form policies (like DROP and TA-DROP) rather than relying on repeated fluid re-solving. For broader surveys on deterministic pricing and their limitations, see \citet{BC} and \citet{GT}.

The study most closely related to ours is the contemporaneous work by \citet{abdallah2025dynamic,abdallah2026regime}, who also analyze dynamic pricing with limited inventory in high-demand regimes. Their model shares key similarities with our setting, but differs in that it addresses uncertainty in buyer arrival rates, whereas our framework benefits from a known and fixed pool of potential buyers—reflecting the particular context of hotel upgrades. By leveraging this demand predictability, our analysis yields closed-form pricing solutions without the need to solve complex Hamilton-Jacobi-Bellman (HJB) equations. Furthermore, utilizing these explicit closed forms enables us to establish strictly stronger performance bounds.

Related but distinct lines of research include \citet{reiman2008asymptotically}, who show that single re-optimization suffices to achieve sublinear regret in network revenue management problems, and \citet{mcafee2008dynamic}, who derive the exact optimal dynamic pricing policy for demand models with constant elasticity. This constant elasticity model directly relates to the heavy-tailed Pareto distribution we study in \Cref{sec:pareto}. However, while \citet{mcafee2008dynamic} focus strictly on characterizing the optimal pricing path, our work investigates the vulnerability of practical run-out heuristics under these heavy-tailed conditions and introduces the TA-DROP algorithm to restore asymptotic optimality.

\subsection{Main Contributions and Insights}
In this paper, we show that simple, interpretable dynamic pricing algorithms approach optimal revenue as the number of potential buyers grows across distinct classes of willingness-to-pay (WTP) distributions. Specifically, our work provides both conceptual insights and rigorous technical bounds, detailed as follows:

\begin{enumerate}
\item \textbf{The Efficacy of Simple Pricing Heuristics:} As a primary conceptual contribution, we demonstrate that highly complex optimization models are not strictly necessary for maximizing revenue in large markets. Simple, rule-based algorithms that adjust prices based solely on current inventory and anticipated future demand can achieve absolute asymptotic optimality.

\item \textbf{Tail-Aware Pricing for Heavy-Tailed Demand:} We establish how the tail behavior of customer valuations dictates the fundamental structure of optimal pricing heuristics. Under heavy-tailed distributions, matching the standard run-out rate severely underprices the option value of extreme valuations. To achieve near-optimality, sellers must intentionally inflate prices by restricting the target probability of a sale to a value strictly below the standard fluid run-out rate.

\item \textbf{Convergence under Exponential WTP:} On the technical front, we mathematically prove that under exponentially distributed WTP, a simple pricing algorithm (DROP)—which sets a price such that the next customer buys with a probability equal to the remaining inventory divided by remaining customers—converges to the optimal revenue with a gap of $o(n^{\epsilon-1})$ for all $\epsilon > 0$.

\item \textbf{Convergence under Pareto WTP:} We prove that under heavy-tailed Pareto distributed WTP with shape parameter $\beta > 1$, an augmented pricing algorithm (TA-DROP)—which intentionally inflates the price such that the sale probability becomes a regulated function of the standard run-out rate—converges to the optimal revenue with an optimality gap of $o(n^{\epsilon+\frac{1}{\beta}-1})$ for all $\epsilon>0$.
\end{enumerate}

\subsection{Paper Organization}
In \Cref{sec:model}, we formalize the model and provide a detailed presentation of the DROP algorithm. In \Cref{sec:mainResults}, we establish that, assuming a large customer base, DROP achieves near-optimal revenue for exponentially distributed utilities. Furthermore, in \Cref{sec:numerical}, we evaluate the finite-sample performance and optimality gaps of these algorithms using numerical simulations. Finally, in \Cref{sec:pareto}, we extend our analysis to capture the heavy-tailed nature of the Pareto distribution, introduce the TA-DROP heuristic, and establish its strict polynomial convergence rate.
We will outline the main steps of the proofs in the main body of this document; however, the formal arguments and detailed derivations are deferred to \Cref{sec:proofs}. 

\section{Preliminaries}\label{sec:model}
A monopolist seller has an initial inventory of $K$ units of a single item and  $N$ potential customers, each with their own random utility (or willingness-to-pay) for the item. The buyers are homogeneous and arrive sequentially, each observing their i.i.d utility $v_i$, drawn from a fixed utility distribution with continuous CDF $F(\cdot)$. After observing the price $p_i$, the buyer $i$ decides to buy the item if $v_i \geq p_i$ and leaves without making a purchase otherwise. This implies that by setting price $p_i$ a purchase happens with probability $1 - F(p_i)$. 

The seller's problem is to determine the optimal prices $(p_i)_{i\in[N]}$ that maximize total expected revenue. This optimization can be formulated as a dynamic programming problem, characterized by a state space of $\mathbb{N} \times \mathbb{N}$ and defined by the following Bellman equation for the value function $V(\cdot\;,\cdot)$:
\begin{align}\label{eq:bellman}\tag{Bellman equation}
V(\inventory,\customer) = \max_p F(p)V(\inventory,\customer-1) + \left(1-F(p)\right)\left(V(\inventory-1,\customer-1) + p\right),
\end{align}
where the first term reflects the scenario in which the buyer does not make a purchase, and the second term is for the case when a purchase is made. Hence the seller is maximizing $V(K,N)$ with the boundary condition $V(\cdot,0) = V(0,\cdot) = 0$. 

Throughout the paper we use $k,n$ as the state variable of the dynamic programming where $k$ is defined as the remaining inventory and $n$ indicates the number of remaining potential customers. It is important to note that defining the state as a pair $(\inventory,\customer)$ is valid as the seller has complete information about the inventory level and future arrivals.\footnote{In the context of our hotel upgrade application, this formulation assumes there are no same-day reservations or no-shows, ensuring the total pool of eligible customers is fixed and known in advance.} Our proposed dynamic pricing algorithm, called \emph{Dynamic Run-Out Pricing} (DROP), sets the price in state $(\inventory,\customer)$ such that the probability of the next customer making a purchase is equal to $\frac{\inventory}{\customer}$.  This contrasts with \emph{Static Run-Out Pricing} (SROP), introduced by \citet{Gvr}, where a fixed price is chosen using the same principle at the outset and remains unchanged throughout the entire horizon. This approach is grounded in the fluid approximation of the problem, which involves solving the following linear programming (LP) relaxation of the original discrete problem:
\begin{align*}\label{LP}\tag{LP}
    \max_p p \cdot \min \left(k,n (1-F(p))\right)
\end{align*}
This LP relaxation renders the problem analytically tractable by replacing discrete, binary purchase decisions with deterministic fractional allocations, where a price $p$ yields exactly $1-F(p)$ units of sale. Solving this deterministic fluid model establishes a baseline strategy for maximizing expected revenue. By anchoring our dynamic algorithm to the structural insights of this relaxation, we effectively balance the immediate reward of a sale against the opportunity cost of depleting inventory, ensuring optimal resource allocation throughout the horizon.

To formally characterize this baseline pricing strategy, we must examine the optimal solution to the fluid approximation in \eqref{LP}. In this deterministic regime, the optimal approach simply equates the total expected demand to the available capacity, satisfying $n(1-F(p)) = k$. Isolating the optimal price $p$ from this equality requires inverting the demand probability, which naturally motivates the following definition:

\begin{definition}\label{def:qunatileDef}
For a distribution with a continuous cumulative distribution function (CDF) $F(\cdot)$, define function $\bar{q}_F: [0,1] \rightarrow \mathbb{R}$ to be the inverse of the function $1-F(\cdot)$. That is,
\begin{align*}
    \bar{q}_F(x) = p \iff 1-F(p) = x
\end{align*}
\end{definition}
In other words, if $F(\cdot)$ represents the distribution of a customer's willingness to pay, then $\bar{q}_F(x)$ outputs the exact price required to secure a purchase with a target probability of $x$. 

Before proceeding, we introduce the following proposition, which serves as a benchmark to compare the performance of our algorithms.
\begin{restatable}{proposition}
{propSROP}\label{propSROP}
    The Static run-out pricing (SROP) introduced in \citet{Gvr} sells asymptotically $K-e^{-K}\frac{K^K}{(K-1)!}$ items on average as the total number of buyers $N$ approaches infinity.
\end{restatable}
\begin{corollary}
    The expected revenue of the SROP is asymptotically equal to
    \begin{align*}
        \bar{q}_F\left(\frac{K}{N}\right)\left(K-e^{-K}\frac{K^K}{(K-1)!}\right)
    \end{align*}
\end{corollary}

The DROP algorithm dynamically adapts this static fluid solution. By substituting the instantaneous state variables—remaining inventory $k$ and remaining potential customers $n$—into our target probability, DROP determines the price for the very next arrival as $p = \bar{q}_F\left(\frac{k}{n}\right)$. This ensures that the seller continually adjusts the price along the quantile function to match the current ratio of supply to expected demand. Consequently, the recursion for updating the price can be expressed as follows:
\begin{align}\label{DICRec}\tag{DROP recursion}
    D(\inventory,\customer) = 
    \left(1-\frac{\inventory}{\customer} \right)
    D(\inventory,\customer-1) + \frac{\inventory}{\customer}\left( D(\inventory-1,\customer-1) + \bar{q}_F\left(\frac{\inventory}{\customer}\right)\right)
\end{align}
It is noteworthy that DROP is less effective in scenarios with small $n$. In particular, when $n=k$, DROP sets the price to $0$ since the probability of purchase is supposed to be $1$. Given that this paper focuses on high-demand conditions (large $N$), setting a near-optimal price for situations where $\frac{k}{n}$ is large is not crucial for convergence. Such cases have a low probability of occurring, and their effect will cancel out in the overall analysis. Later, in \Cref{sec:numerical}, we will refine the algorithm by introducing DROP* to better address these lower values.

In the remainder of this section, we establish a straightforward yet foundational result: under bounded distributions, while SROP suffers from significant inventory excess due to its rigid pricing structure, our dynamic adaptation (DROP) achieves asymptotically optimal revenue.

\begin{definition}\label{def:bounded}
    A willingness-to-pay distribution is bounded if there exists a maximum valuation $\bar{p} < \infty$ such that no customer is willing to make a purchase at a price strictly greater than $\bar{p}$; equivalently, $F(\bar{p}) = 1$.
\end{definition}
 
\begin{restatable}{theorem}{thmUniform}\label{thmUniform}
    For any bounded willingness-to-pay distribution $F(\cdot)$, the optimality gap under the DROP algorithm vanishes as the market size grows:
    \begin{align*}
    \lim_{n \rightarrow \infty} \left[V(k,n) - D(k,n)\right] = 0.
    \end{align*}
\end{restatable}

To prove this theorem, let $\bar{p}$ denote the upper bound of the distribution $F(\cdot)$. Because no single item can ever be sold for a price exceeding $\bar{p}$, the total expected revenue is strictly bounded above by $\bar{p}k$. Furthermore, since the optimal policy always weakly outperforms the DROP heuristic, we can establish the following sequence of inequalities:
$$D(k,n) \leq V(k,n) \leq \bar{p}k.$$
By the Squeeze Theorem, if we can demonstrate that the expected revenue under DROP approaches this theoretical maximum—that is, $\lim_{n \rightarrow \infty} D(k,n) = \bar{p}k$—then the optimality gap must vanish. Consequently, the proof of \Cref{thmUniform} directly reduces to establishing the following proposition.

\begin{restatable}{proposition}{uniformD}\label{propUniform}
For any bounded willingness-to-pay distribution $F(\cdot)$,
    \begin{align*}
    \lim_{n \rightarrow \infty} D(k,n) = \bar{p}k.
    \end{align*}
\end{restatable}

\section{Asymptotic Optimality under Exponential Demand}\label{sec:mainResults}
In this section, and specifically in \Cref{thmExp}, we demonstrate that the performance guarantees of the DROP algorithm are not restricted to bounded distributions. Instead, we establish that these guarantees extend to unbounded cases, achieving strict asymptotic convergence under an exponentially distributed willingness-to-pay. \Cref{table:results} summarizes the closed-form asymptotic revenue limits for each pricing benchmark as a function of the total number of potential customers $N$ and the initial inventory $K$.
\begin{table}[htbp]
\centering
\small

\begin{tabular}{|c|c|c|}
\hline
& \textbf{Bound} & \textbf{Result} \\
\hline
\text{SROP} & $\left(\log(N)-\log(K)\right)\left(K-e^{-K}\frac{K^K}{(K-1)!}\right)$ & \Cref{propSROP}  \\
\hline
\text{DROP} & $K\log(N)-K - \log(K!)$ & \Cref{DConverging}\\
\hline
\text{OPT Online} & $K\log(N)-K - \log(K!)$ & \Cref{vConverging}\\
\hline
\text{OPT Offline} & $KH_N - KH_{K-1} + K-1$ & \Cref{PConverging} \\
\hline
\end{tabular}
\caption{Summary of convergence characteristics proven in this paper under the exponential distribution. $K$ represents the initial inventory and $N$ denotes the total number of potential customers.}
\label{table:results}
\end{table}

\begin{definition}\label{def:exponential}
    A willingness to pay distribution is defined as exponential if its CDF takes the form $F(p) = 1-e^{-cp}$ for some constant $c$.
\end{definition}
\begin{remark}
    For any exponential distribution, one can assume $c=1$ without loss of generality by changing the currency. Thus, for simplicity, we will assume $c=1$ when discussing exponential distributions, leading to the form $F(p) = 1 - e^{-p}$.
\end{remark}

In the proof of \Cref{thmUniform}, we relied on the existence of an upper bound for the revenue and demonstrated that this bound is achieved by DROP when the number of customers becomes large. However, this approach cannot be applied in the case of exponential distributions, as the distribution is unbounded above. Since the result in \Cref{thmUniform} may be expected, it is not necessarily a strong indicator of DROP's effectiveness, as any sufficiently smart algorithm should be able to leverage a large number of customers to achieve optimal revenue in bounded cases. By contrast, the convergence in \Cref{thmExp} is indeed surprising and highlights the true power of DROP in handling unbounded distributions.

\begin{restatable}{theorem}{thmExp}\label{thmExp}
    For an exponential willingness to pay distribution $F(.)$, and any $\alpha < 1$
    \begin{align*}
    \lim_{\customer \rightarrow \infty} \left[ V(\inventory,\customer)- D(\inventory,\customer)\right] n^{\alpha}
    = 0
    \end{align*}
\end{restatable}
Intuitively, the proof of this theorem relies on showing that both expressions  $D(k,n)$ (in \Cref{DConverging}) and $V(k,n)$ (in \Cref{vConverging}) converge to the same limit.

\begin{restatable}{proposition}{propD}\label{DConverging} For an exponential willingness to pay distribution $F(.)$, and any $\alpha < 1$
    \begin{align*}
        \lim_{\customer \rightarrow \infty} \left[\inventory\log(\customer) - \inventory - \log(\inventory !) - D(\inventory,\customer) \right]n^{\alpha} = 0
    \end{align*}
\end{restatable}

\begin{restatable}{proposition}{propV}\label{vConverging} For an exponential willingness to pay distribution $F(.)$, and any $\alpha < 1$
    \begin{align*}
        \lim_{\customer \rightarrow \infty} \left[ \inventory\log(\customer) - \inventory - \log(\inventory !) - V(\inventory,\customer) \right]n^{\alpha} = 0
    \end{align*}
\end{restatable}

Thus, since both terms converge to the same value, namely $0$, their difference also vanishes. This establishes that $[D(k,n)-V(k,n)]n^\alpha$ approaches zero as $n$ approaches infinity, completing the proof of \Cref{thmExp} and highlighting the efficiency of DROP under exponential willingness-to-pay distributions. Furthermore, in \Cref{remark:optimalPrices} in the Appendix we show that the prices posted by DROP are also converging to the optimal prices.

Crucially, the theoretical guarantees established here are strictly stronger than standard multiplicative convergence. While asymptotic optimality is frequently defined by a performance ratio converging to 1, our analysis demonstrates a vanishing absolute gap (commonly referred to as regret). We show that the absolute difference between the algorithmic and optimal revenues not only shrinks to 0, but decays at a strict polynomial rate, falling faster than $n^{-\alpha}$ for any $\alpha < 1$.

A natural question arises regarding the performance of an offline algorithm, one that has full knowledge of all willingness-to-pay values beforehand, in this setting. While $V(k,n)$
denotes the revenue generated by the optimal online algorithm, our focus is on determining the maximum revenue that can be achieved by the optimal offline algorithm, often referred to as the \emph{prophet} in the theoretical Computer Science literature. 

\begin{restatable}{proposition}{propP}\label{PConverging}
    The expected revenue of the optimal offline algorithm is: 
    \begin{align*}
        kH_n - kH_{k-1} + k-1.
    \end{align*}
    Where $H_n$ is the $n$-th harmonic number and is defined as the sum of the reciprocals of the first $n$ positive integers. It can be expressed mathematically as:
    \begin{align*}
        H_n = \sum_{i=1}^n \frac{1}{i}
    \end{align*}
\end{restatable}
\begin{proof}
    The revenue of the optimal offline corresponds to the expected value of the $k$ highest values among $n$ samples drawn from an exponential distribution. According to \citet{renyi1953theory}, the expected value of the i-th order statistic from an exponential distribution is $H_n - H_{n-i}$. A brief proof outline can be constructed using the memorylessness property of the exponential distribution. Specifically, assuming  $x_n>\cdots>x_2>x_1$, we have:
    \begin{align*}
        \EX [x_{i}] &= \EX\left[\EX [x_{i}|x_{i-1}]\right] = \EX\left[\EX [x_{i} - x_{i-1}|x_{i-1}]+x_{i-1}\right]\\
        &= \frac{1}{n+1-i} + \EX\left[x_{i-1}\right] = \sum_{j=n+1-i}^n \frac{1}{j} = H_n - H_{n-i}
    \end{align*}
    Thus, the revenue of the optimal offline can be expressed as:
    \begin{align*}
        \sum_{i=0}^{k-1} \left(H_n - H_i\right) &= kH_n - \sum_{i=1}^{k-1} H_i = kH_n - \sum_{i=1}^{k-1} \frac{k-i}{i} \\
        &= kH_n - \sum_{i=1}^{k-1} \frac{k}{i} + k-1 = kH_n - kH_{k-1} + k-1
    \end{align*}
    This gives the total expected revenue for the optimal offline.$\Box$
\end{proof}
Notice that $H_n$ is closely related to $\log(n)$. In fact, it is well known that $\lim_{n \rightarrow \infty} H_n-\log(n) = \gamma$, where $\gamma \approx 0.577$ is the \textit{Euler–Mascheroni} constant. This implies that while $\lim_{n\rightarrow \infty} \frac{\text{DROP}(k,n)}{\text{OFF}(k,n)} = 1$, we still have: $$\lim_{n\rightarrow \infty} \text{OFF}(k,n) - \text{DROP}(k,n) = k\gamma - kH_{k-1}+2k-1+\log(k!) > 0$$ 
This differs from the bounded willingness-to-pay case (\Cref{thmUniform}), where DROP generated revenue equivalent to both the optimal online and optimal offline algorithms. In the exponential case, however, while DROP converges to the optimal online strategy, it is not optimal when all willingness-to-pay values are known. This indicates that, unlike in the bounded case, having information about customers' willingness to pay is valuable for the seller. This is mainly an artifact of the constant gap between the optimal online ($\approx$ DROP) and the optimal offline strategies which is given by $k\gamma - kH_{k-1}+2k-1+\log(k!)$. For a visual representation of the gap, refer to \Cref{fig:Harmonic}.
\begin{figure}[h!]
\begin{center}
\input{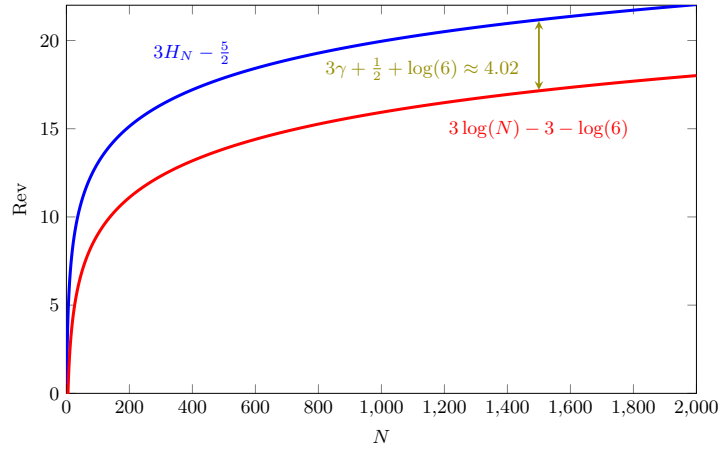}
\end{center}
\caption{The gap between $kH_n - kH_{k-1} + k-1$ (blue) and $k \log(n)-k - \log(k!)$ (in red) for $k=3$.}
\label{fig:Harmonic}
\end{figure}

Finally, the following table summarizes a comparison between our convergence bounds under exponential demand and those established in prior literature. Notably, works such as \citet{abdallah2026regime} and \citet{reiman2008asymptotically} consider more general problem settings, which lead to weaker convergence rates than those we establish in this paper. To quantify the rate of convergence, we define it as the largest value of $\alpha$ such that:
\begin{align*}
    \lim_{n\rightarrow \infty}  [V(n)-D(n)]n^{\alpha} = 0.
\end{align*}
where $V(n)$ denotes the optimal revenue and $D(n)$ represents the revenue achieved by the proposed algorithm. A rate $\alpha > -1$ corresponds to sublinear regret.

\begin{table}[htbp]
\centering
\small
\renewcommand{\arraystretch}{1.5}
\begin{tabular}{|c|c|}
\hline
Paper &  \textbf{$\alpha$}\\
\hline
 \citet{reiman2008asymptotically} & $-\frac{1}{2}$ \\
\hline
\citet{Gvr} & $-\epsilon$ \\
\hline
\citet{abdallah2026regime} & $0$ \\
\hline
This Paper & $1-\epsilon$ \\
\hline
\end{tabular}
\caption{Comparison of convergence rates across different works for the case with exponential demand. The rate $\alpha$ is defined such that $\lim_{n \to \infty} [V(n) - D(n)]n^{\alpha} = 0$.}
\label{tab:convergence-rates}
\end{table}

\subsection{Numerical Experiments and Small-Market Refinements }\label{sec:numerical}
In this subsection, we demonstrate the robustness of our approach in finite-sample settings, showing that a simple algorithmic refinement to DROP yields strong performance even when the market size $N$ is small.

As mentioned at the end of Section $\ref{sec:model}$, the prices set by DROP are not meaningful for small $n$ under certain regimes, including uniform and exponential WTP. To address this issue, we introduce a refined version of the DROP algorithm, called DROP*. Let $p^M$ represent the monopoly price for one item and one buyer defined as $p^M = \argmax_p p(1-F(p))$. DROP* ensures that the prices it sets are not below the monopoly price. More specifically, if DROP sets the price $p$ in a given state, DROP* will set the price $\max(p,p^M)$ instead.
\begin{example}
    For a willingness-to-pay uniformly distributed over $[0,\bar{p}]$, the optimal monopoly price is $p^M = \frac{\bar{p}}{2}$. In the case of an exponentially distributed willingness-to-pay, the optimal price is  $p^M = 1$.
\end{example}
In \Cref{fig:uniform} we compare various pricing strategies under the uniform distribution. To enhance clarity, we limit the analysis to $n \leq 100$. As benchmarks, the performance of the monopoly pricing $p^*=0.5$ is represented by the red curve, while SROP is depicted in green.
% \begin{figure}
%     \centering
%     \include{Figures/uniform}
%     \caption{The optimality gap of DROP (in blue) and DROP*(in black) for $k=5$ and willingness to pay $\sim U[0,1]$ is shown. The largest gap occurs at $n = 5$ with a value of $0.6$ for DROP and $0.1$ for DROP*. The monopoly pricing gap approaches $2.5$ as all 5 products are sold at a price of $0.5$.}
%     \label{fig:uniform}
% \end{figure}
In \Cref{fig:exponential} the gap between each strategy and the optimal online performance is illustrated under the assumption of an exponential distribution.
% \begin{figure}
%     \centering
%     \include{Figures/exponential}
%     \caption{The optimality gap of DROP (in blue) and DROP*(in black) for $k=5$ and willingness to pay $\sim exp(1)$ is illustrated. The largest gap value is  $0.73$ for DROP and $0.22$ for DROP*. As a benchmark, the performance of SROP is represented in green.}
%     \label{fig:exponential}
% \end{figure}

\begin{figure}[h!]
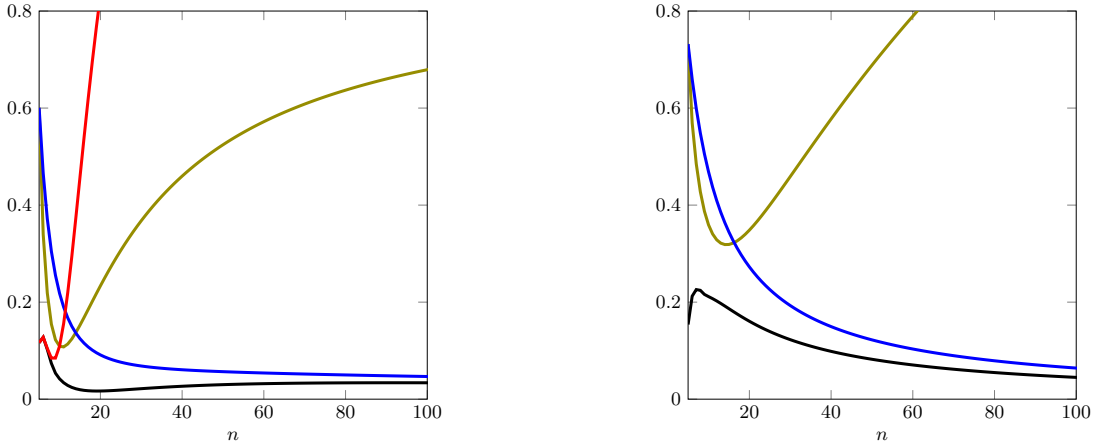

    \centering
    \begin{subfigure}[t]{0.48\textwidth}
        \centering
        \include{Figures/uniform}
        \caption{The optimality gap of DROP (in blue) and DROP* (in black) for $k=5$ with willingness to pay $\sim U[0,1]$. The largest gap occurs at $n = 5$, with a value of $0.6$ for DROP and $0.1$ for DROP*. The monopoly pricing gap approaches $2.5$ as all 5 products are sold at price $0.5$.}
        \label{fig:uniform}
    \end{subfigure}
    \hfill
    \begin{subfigure}[t]{0.48\textwidth}
        \centering
        \include{Figures/exponential}
        \caption{The optimality gap of DROP (in blue) and DROP* (in black) for $k=5$ with willingness to pay $\sim \exp(1)$. The largest gap is $0.73$ for DROP and $0.22$ for DROP*. SROP performance is shown in green.}
        \label{fig:exponential}
    \end{subfigure}
    \caption{Comparison of optimality gaps for DROP, DROP*, and SROP under two demand distributions.}
\end{figure}

\section{Tail-Aware Pricing for Heavy-Tailed Demand } \label{sec:pareto}
In this section, we investigate the vulnerability of standard run-out rate heuristics when faced with heavy-tailed valuation distributions. We demonstrate that under a heavy-tailed regime, the standard DROP algorithm fails to achieve asymptotic optimality because it does not capitalize on the outsized revenue potential of extreme valuations at the tail. To overcome this limitation, we introduce the Tail-Aware DROP (TA-DROP) algorithm, an augmented heuristic that intentionally inflates prices to match the optimal online policy and restore strict convergence.

We model customer willingness-to-pay using the Pareto distribution. The Pareto distribution is a particularly good choice for this analysis because its structural properties yield clean, closed-form expressions for both the optimal dynamic policy and our heuristic approximations.

\begin{definition}\label{def:Pareto}
    A willingness-to-pay distribution follows a Pareto distribution if its CDF takes the form $F(p) = 1 - \left(\frac{c}{p}\right)^\beta$ for all $p \ge c$, where $c > 0$ is the scale parameter and $\beta > 0$ is the shape parameter.
\end{definition}

\begin{remark}
    This Pareto formulation corresponds exactly to the constant elasticity demand model analyzed by \citet{mcafee2008dynamic}. Under this distribution, the shape parameter $\beta$ dictates the price elasticity of demand, which remains strictly constant across all price levels.
\end{remark}

Similar to the exponential case, without loss of generality, we can normalize the scale parameter to $c=1$ by rescaling the currency units. Therefore, for the remainder of this analysis, we assume $F(p) = 1 - p^{-\beta}$ for $p \ge 1$. Furthermore, we restrict our attention to the regime where $\beta > 1$. When $\beta \leq 1$, the distribution has an infinite mean, implying that the optimal strategy would trivially degenerate to setting an infinite price to extract infinite expected revenue. 

\subsection{Optimal Dynamic Pricing}
To find the optimal price under Pareto WTP, we define the marginal opportunity cost of inventory as $\Delta V = V(k, n-1) - V(k-1, n-1)$. By taking the derivative of the Bellman equation \eqref{eq:bellman} with respect to price $p$, we obtain the first-order condition:
\begin{equation*}
    p^{-\beta} - \beta p^{-\beta-1} \left( p - \Delta V \right) = 0
\end{equation*}
Solving this equation for $p$ isolates the optimal pricing policy, revealing a simple fractional markup over the opportunity cost of the inventory:
\begin{equation*}
    p^* = \frac{\beta}{\beta - 1} \Delta V
\end{equation*}

Substituting this optimal price back into the Bellman recursion allows us to express the value function entirely in terms of the marginal value of inventory:
\begin{align}\label{eq:diffeq}
V(k,n) &= (1-p^{-\beta})V(k,n-1) + p^{-\beta}\left(V(k-1,n-1) + \frac{\beta}{\beta - 1} \Delta V\right) \nonumber\\
&= V(k,n-1) + p^{-\beta}\left(\frac{\beta}{\beta - 1} \Delta V - \Delta V \right) \nonumber\\
&= V(k,n-1) + \left(\frac{\beta}{\beta - 1} \Delta V\right)^{-\beta} \left(\frac{\Delta V}{\beta-1}\right) \nonumber\\
&= V(k,n-1) + \frac{\left(\beta - 1\right)^{\beta-1}}{\beta^\beta}(\Delta V)^{1-\beta}
\end{align}

To build intuition for the asymptotic behavior, consider the continuous approximation for a single unit of inventory ($k=1$). By substituting $V(1, n) - V(1, n-1) \approx dV/dn$ and $\Delta V = V$, we obtain the differential equation:
\begin{equation*}
    \beta V^{\beta-1} dV \approx \left(1 - \frac{1}{\beta}\right)^{\beta-1} dn
\end{equation*}
Integrating this yields the asymptotic structural scaling of the value function and, consequently, the optimal price:
\begin{align}\label{eq:optPrice}
    V(1,n) \approx n^{\frac{1}{\beta}}\left(1 - \frac{1}{\beta}\right)^{1-\frac{1}{\beta}} \quad \implies \quad p^* \approx n^{\frac{1}{\beta}}\left(1 - \frac{1}{\beta}\right)^{-\frac{1}{\beta}}
\end{align}
More generally, by assuming that in the asymptotic regime $V(k,n) \approx C_k n^{\frac{1}{\beta}}$ with $C_0 = 0$, from \eqref{eq:diffeq} we get:
\begin{align*}
    \frac{d C_k n^{\frac{1}{\beta}}}{dn} = \frac{C_k}{\beta} n^{\frac{1}{\beta}-1} =  \frac{\left(\beta - 1\right)^{\beta-1}}{\beta^\beta}\left((C_k-C_{k-1}) n^{\frac{1}{\beta}}\right)^{1-\beta}.
\end{align*}
This gives a non-linear recursive equation for $C_k$:
\begin{align}\label{eq:cRecursion}
    &C_k = \left((C_k-C_{k-1})\frac{\beta}{\beta-1}\right)^{1-\beta}\nonumber\\ \iff& C_{k-1} = C_k-\frac{\beta-1}{\beta}C_k^{\frac{1}{1-\beta}}
\end{align}
Note that the right-hand side is increasing in $C_k$ and for $C_k = C_{k-1}$ the left-hand side will be greater. Furthermore, for $C_k = \infty$ the right-hand side will approach infinity, therefore there exists a unique solution to the above equation with $C_k >C_{k-1}$ as one expects.
\begin{remark}\label{rmk:ck<k}
    $\lim_{k\rightarrow \infty} C_k - k^{1-\frac{1}{\beta}} = 0$ which means $C_k$ is always less than $k$. (See Theorem 2 at \cite{mcafee2008dynamic})
\end{remark}

It should be noted that the derivations in this section rely on continuous approximations and are primarily heuristic in nature. Rather than establishing rigorous performance bounds, this exercise serves to build structural intuition regarding the optimal pricing trajectory and to provide a qualitative benchmark against the DROP algorithm. Due to the analytical intractability of the general multi-item setting, we restrict our focus in the subsequent subsection to the case $k=1$, which allows for a precise, closed-form comparison between the DROP prices and this (approximately) derived optimal price.

\subsection{DROP for $k=1$}
In this subsection, we demonstrate the shortcomings of DROP under heavy-tailed demand. Note that under the standard Dynamic Run-Out Pricing (DROP) heuristic, the price is dynamically set such that the probability of a sale exactly matches the remaining run-out rate of the inventory. Setting this condition yields $1 - F(p) = p^{-\beta} = \frac{k}{n}$. Solving for $p$ provides a beautifully compact, closed-form expression for the policy:
\begin{equation*}
    p^{DROP} = \left(\frac{n}{k}\right)^{\frac{1}{\beta}}
\end{equation*}
Therefore, the recursion for $D(1,n)$ will look like:
\begin{align*}
    D(1,\customer) =& 
    \left(1-\frac{1}{\customer} \right)
    D(1,\customer-1) + \frac{1}{\customer}\left( D(0,\customer-1) + n^{\frac{1}{\beta}}\right)\\
    =& \frac{\sum_{i=1}^n i^{\frac{1}{\beta}}}{n}=\frac{\beta}{1+\beta}n^{\frac{1}{\beta}}.
\end{align*}
This value, however, is strictly less than the theoretical benchmark $n^{\frac{1}{\beta}}\left(1 - \frac{1}{\beta}\right)^{1-\frac{1}{\beta}}$ derived in \eqref{eq:optPrice}.\footnote{This strict inequality follows from straightforward algebraic rearrangement. By substituting $x = \beta-1$, the relationship is established by invoking the property that the function $\left(1-\frac{1}{x+1}\right)^x$ is strictly decreasing for all $x>0$.} This indicates that, in sharp contrast to the bounded and light-tailed regimes, the standard DROP algorithm fails to achieve a vanishing optimality gap. 

A direct comparison between $p^{\text{DROP}}$ and our derived optimal price $p^*$ for $k=1$ in \eqref{eq:optPrice} immediately reveals that the standard DROP heuristic severely underprices the inventory. Because the Pareto distribution exhibits a heavy right tail, the probability of securing a sale at inflated prices decays much more gradually than it would under an exponential distribution. By rigidly matching the standard fluid run-out rate, the DROP heuristic ignores the immense option value embedded in waiting for extreme valuations.

\subsection{TA-DROP Heuristic}

To correct this, we introduce the \textbf{Tail-Aware DROP (TA-DROP)}. Rather than matching the exact run-out rate, using \eqref{eq:optPrice}, TA-DROP intentionally inflates the price by a factor :
\begin{align*}
    p^{TA} = \frac{\beta}{\beta - 1} \left(C_k-C_{k-1}\right)n^{\frac{1}{\beta}} = C_k^{\frac{1}{1-\beta}}n^{\frac{1}{\beta}}.
\end{align*}
Or equivalently TA-DROP restricts the target sale probability by a factor of $\frac{\beta-1}{\beta}$:
\begin{equation}
    1 - F(p^{TA}) = \frac{C_k^{\frac{\beta}{\beta-1}}}{n}.
\end{equation}
By artificially lowering the target sale probability, TA-DROP forces the platform to quote higher prices, protecting the inventory for high-valuation arrivals. Note that for extremely heavy tails ($\beta \to 1^+$), the multiplier approaches zero, causing the heuristic to become hyper-conservative and set prices far above the standard run-out pricing. 

This tail-aware adjustment is sufficient to restore asymptotic optimality with a strict polynomial convergence rate, as formalized in the following theorem.

\begin{restatable}{theorem}{thmPareto}\label{thmPareto}
    For a Pareto willingness-to-pay distribution $F(p)$ with shape parameter $\beta > 1$, and for any convergence exponent $\alpha < 1-\frac{1}{\beta}$, the difference between the optimal online policy and the TA-DROP heuristic satisfies:
    \begin{align*}
    \lim_{n \rightarrow \infty} \left[V(k,n) - D_{TA}(k,n)\right] n^{\alpha} = 0
    \end{align*}
\end{restatable}

Similar to \Cref{thmExp}, we find it more tractable to prove that both the TA-DROP and the optimal revenues converge to the same asymptotic limit $C_kn^{\frac{1}{\beta}}$.
\begin{restatable}{proposition}{propParetoD}\label{propParetoD} For a Pareto willingness to pay distribution $F(.)$ with shape parameter $\beta > 1$, and any $\alpha < 1-\frac{1}{\beta}$ we have:
    \begin{align*}
        \lim_{\customer \rightarrow \infty} \left[ C_kn^{\frac{1}{\beta}} - D_{TA}(\inventory,\customer) \right]n^{\alpha} = 0,
    \end{align*}
    where $C_k$ is defined in \eqref{eq:cRecursion}.
\end{restatable}

\begin{restatable}{proposition}{propParetoV}\label{propParetoV} For a Pareto willingness to pay distribution $F(.)$ with shape parameter $\beta > 1$, and any $\alpha < 1-\frac{1}{\beta}$ we have:
    \begin{align*}
        \lim_{\customer \rightarrow \infty} \left[ C_kn^{\frac{1}{\beta}} - V(\inventory,\customer) \right]n^{\alpha} = 0,
    \end{align*}
    where $C_k$ is defined in \eqref{eq:cRecursion}.
\end{restatable}
Note that SROP policy sells asymptotically $K-e^{-K}\frac{K^K}{(K-1)!}$ items with the price of $q\left(\frac{K}{N}\right) = \left(\frac{N}{K}\right)^{\frac{1}{\beta}}$. Therefore, the gap between SROP and optimal will be unbounded:
\begin{align*}
    V(K,N) - S(K,N) = \left(C_K- K^{1-\frac{1}{\beta}} +  e^{-K}\frac{K^{K-\frac{1}{\beta}}}{(K-1)!} \right)N^{\frac{1}{\beta}} 
\end{align*}
Note that while $C_K-K^{1-\frac{1}{\beta}}$ converges to 0 as $K$ increases, the term $e^{-K}\frac{K^{K-\frac{1}{\beta}}}{(K-1)!}$ persists even in the fluid regime. See \Cref{fig:pareto} for a visual illustration.

% \section{Kumaraswamy distribution}
% In this section we look at Kumaraswamy distribution with $a = 1$ for tractibility. Therefore the optimal price will be to set $p^* = \frac{b\Delta+1}{b+1}$ and the probability of making a purchase will be $\left(\frac{b(1-\Delta)}{b+1}\right)^b$. This means:
% \begin{align*}
%     V(k,n) =& V(k,n-1) + \left(\frac{b(1-\Delta)}{b+1}\right)^b\left(V(k-1,n-1)-V(k,n-1) + \frac{b\Delta+1}{b+1}\right)\\
%     =&V(k,n-1) + \left(\frac{b(1+V(k-1,n-1)-V(k,n-1))}{b+1}\right)^{b+1},
% \end{align*}
% therefore, $V(1,n) \approx 1-n^{-\frac{1}{b}}$, therefore the price will be approximately $1-\frac{b}{b+1}n^{-\frac{1}{b}}$. The DROP however set the price so that $(1-p)^b = \frac{k}{n}$ or $p=1-k^{\frac{1}{b}}n^{-\frac{1}{b}}$.

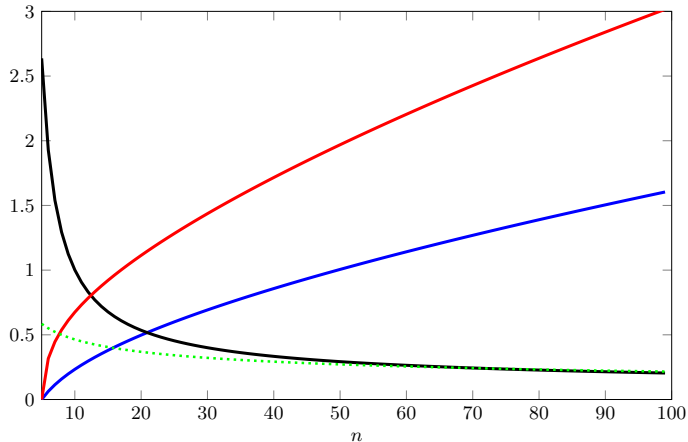
\begin{figure}[h!]
    \centering
    \begin{tikzpicture}[scale=0.8, transform shape]
\begin{axis}[
      width=12cm,
      height=8cm,
      xlabel={$n$},
      ylabel={},
      legend pos=south east,
      xmin=5, xmax=100,
      ymin=0, ymax=3,
    ]
    
\addplot[color=blue, line width=0.5mm] coordinates{
( 5 , 0.0 )
( 6 , 0.06331228385257859 )
( 7 , 0.1144262984386204 )
( 8 , 0.15842302145093434 )
( 9 , 0.19760261919186917 )
( 10 , 0.23326032686741538 )
( 11 , 0.2662118450389386 )
( 12 , 0.2970087834938706 )
( 13 , 0.3260434730189523 )
( 14 , 0.3536058147203782 )
( 15 , 0.37991657505111576 )
( 16 , 0.405148046110698 )
( 17 , 0.42943746246289294 )
( 18 , 0.45289604207527745 )
( 19 , 0.47561526886999594 )
( 20 , 0.49767137383632765 )
( 21 , 0.5191286038123728 )
( 22 , 0.5400416530297569 )
( 23 , 0.5604575032857078 )
( 24 , 0.5804168380331607 )
( 25 , 0.5999551440186011 )
( 26 , 0.6191035801478257 )
( 27 , 0.6378896704551877 )
( 28 , 0.6563378624297016 )
( 29 , 0.6744699810581825 )
( 30 , 0.6923056012267281 )
( 31 , 0.7098623555714383 )
( 32 , 0.7271561908242834 )
( 33 , 0.744201582715398 )
( 34 , 0.7610117172655375 )
( 35 , 0.777598644622369 )
( 36 , 0.7939734103145035 )
( 37 , 0.8101461678134001 )
( 38 , 0.826126275530612 )
( 39 , 0.841922380781611 )
( 40 , 0.8575424927781405 )
( 41 , 0.8729940463386754 )
( 42 , 0.8882839577093478 )
( 43 , 0.9034186736490462 )
( 44 , 0.9184042147392013 )
( 45 , 0.9332462137221178 )
( 46 , 0.9479499495433998 )
( 47 , 0.9625203776688309 )
( 48 , 0.9769621571591678 )
( 49 , 0.9912796749142281 )
( 50 , 1.0054770674377806 )
( 51 , 1.0195582404244767 )
( 52 , 1.0335268864280245 )
( 53 , 1.0473865008342322 )
( 54 , 1.0611403963326431 )
( 55 , 1.0747917160548575 )
( 56 , 1.088343445526025 )
( 57 , 1.1017984235574048 )
( 58 , 1.1151593521919416 )
( 59 , 1.1284288058012493 )
( 60 , 1.141609239420447 )
( 61 , 1.1547029963972548 )
( 62 , 1.1677123154227687 )
( 63 , 1.1806393370037966 )
( 64 , 1.1934861094298093 )
( 65 , 1.206254594281802 )
( 66 , 1.218946671525245 )
( 67 , 1.2315641442246843 )
( 68 , 1.24410874291382 )
( 69 , 1.2565821296511892 )
( 70 , 1.2689859017886285 )
( 71 , 1.2813215954769923 )
( 72 , 1.2935906889310331 )
( 73 , 1.3057946054734835 )
( 74 , 1.317934716376108 )
( 75 , 1.3300123435141913 )
( 76 , 1.3420287618490754 )
( 77 , 1.3539852017522271 )
( 78 , 1.3658828511829704 )
( 79 , 1.377722857731026 )
( 80 , 1.3895063305339157 )
( 81 , 1.401234342078535 )
( 82 , 1.4129079298952583 )
( 83 , 1.4245280981523933 )
( 84 , 1.4360958191579805 )
( 85 , 1.4476120347754744 )
( 86 , 1.4590776577592344 )
( 87 , 1.4704935730153714 )
( 88 , 1.4818606387929272 )
( 89 , 1.4931796878100663 )
( 90 , 1.504451528319553 )
( 91 , 1.5156769451174803 )
( 92 , 1.5268567004989002 )
( 93 , 1.5379915351637194 )
( 94 , 1.5490821690760477 )
( 95 , 1.5601293022797869 )
( 96 , 1.571133615673297 )
( 97 , 1.582095771745486 )
( 98 , 1.5930164152757484 )
( 99 , 1.6038961739998463 )
};
\addplot[color=black, line width=0.5mm] coordinates{
( 5 , 2.6367422340395716 )
( 6 , 1.9275709092797917 )
( 7 , 1.5395381139566933 )
( 8 , 1.2944250756979887 )
( 9 , 1.125283932915833 )
( 10 , 1.001334395099236 )
( 11 , 0.9064517811046375 )
( 12 , 0.8313742837169666 )
( 13 , 0.7704034071035402 )
( 14 , 0.7198405813162196 )
( 15 , 0.6771803071377018 )
( 16 , 0.6406647019537992 )
( 17 , 0.6090237519374853 )
( 18 , 0.5813168328106251 )
( 19 , 0.5568322536522885 )
( 20 , 0.5350214665255848 )
( 21 , 0.51545474436276 )
( 22 , 0.4977905748793052 )
( 23 , 0.48175406037184665 )
( 24 , 0.46712137532977493 )
( 25 , 0.45370838747444964 )
( 26 , 0.4413621959940528 )
( 27 , 0.4299547496537546 )
( 28 , 0.41937797137534893 )
( 29 , 0.40953998975678196 )
( 30 , 0.40036219472291457 )
( 31 , 0.3917769142040939 )
( 32 , 0.38372556402794267 )
( 33 , 0.3761571621197888 )
( 34 , 0.3690271258581994 )
( 35 , 0.36229629147053544 )
( 36 , 0.355930108990119 )
( 37 , 0.3498979771025823 )
( 38 , 0.34417269026704744 )
( 39 , 0.338729976563215 )
( 40 , 0.3335481093212884 )
( 41 , 0.32860757911808136 )
( 42 , 0.3238908154437148 )
( 43 , 0.31938194945856324 )
( 44 , 0.31506661091546206 )
( 45 , 0.31093175362676817 )
( 46 , 0.30696550488968555 )
( 47 , 0.3031570351081534 )
( 48 , 0.2994964445107833 )
( 49 , 0.2959746643976189 )
( 50 , 0.29258337078064045 )
( 51 , 0.28931490863471865 )
( 52 , 0.28616222526373747 )
( 53 , 0.28311881152304963 )
( 54 , 0.28017864983478447 )
( 55 , 0.2773361680941413 )
( 56 , 0.2745861986995699 )
( 57 , 0.27192394205184556 )
( 58 , 0.2693449339614844 )
( 59 , 0.2668450164828293 )
( 60 , 0.26442031176019754 )
( 61 , 0.26206719852778804 )
( 62 , 0.2597822909531722 )
( 63 , 0.2575624195550503 )
( 64 , 0.2554046139607564 )
( 65 , 0.25330608729906245 )
( 66 , 0.25126422204924737 )
( 67 , 0.24927655718970598 )
( 68 , 0.2473407765083202 )
( 69 , 0.24545469795330632 )
( 70 , 0.24361626391757696 )
( 71 , 0.24182353236208343 )
( 72 , 0.2400746686943549 )
( 73 , 0.23836793832807857 )
( 74 , 0.23670169985752665 )
( 75 , 0.23507439878820335 )
( 76 , 0.23348456177117072 )
( 77 , 0.2319307912942996 )
( 78 , 0.23041176078848125 )
( 79 , 0.22892621011126835 )
( 80 , 0.22747294137423424 )
( 81 , 0.226050815083763 )
( 82 , 0.22465874656794682 )
( 83 , 0.22329570266508725 )
( 84 , 0.22196069865152168 )
( 85 , 0.22065279538878713 )
( 86 , 0.21937109667190668 )
( 87 , 0.21811474676240294 )
( 88 , 0.21688292809111687 )
( 89 , 0.2156748591171791 )
( 90 , 0.21448979233093723 )
( 91 , 0.2133270123895521 )
( 92 , 0.21218583437499916 )
( 93 , 0.21106560216518488 )
( 94 , 0.209965686909662 )
( 95 , 0.20888548560208875 )
( 96 , 0.20782441974237287 )
( 97 , 0.20678193408186019 )
( 98 , 0.2057574954456598 )
( 99 , 0.2047505916265635 )
};

\addplot[color=red, line width=0.5mm] coordinates{
( 5 , 0.0 )
( 6 , 0.31697549148984816 )
( 7 , 0.44621199066202166 )
( 8 , 0.5375312695874097 )
( 9 , 0.6113008015550712 )
( 10 , 0.6748322669905811 )
( 11 , 0.731619000519971 )
( 12 , 0.7836056057602692 )
( 13 , 0.8319860416901985 )
( 14 , 0.8775459882818097 )
( 15 , 0.920830103023162 )
( 16 , 0.9622316892731995 )
( 17 , 1.0020442738509026 )
( 18 , 1.0404929673025478 )
( 19 , 1.0777544105514458 )
( 20 , 1.1139699460775283 )
( 21 , 1.1492545951682924 )
( 22 , 1.183703345487812 )
( 23 , 1.2173956602639908 )
( 24 , 1.2503987799982372 )
( 25 , 1.282770184959535 )
( 26 , 1.3145594622023236 )
( 27 , 1.3458097421598154 )
( 28 , 1.376558818889654 )
( 29 , 1.4068400342822436 )
( 30 , 1.4366829837226707 )
( 31 , 1.466114084990572 )
( 32 , 1.4951570411929254 )
( 33 , 1.5238332207163587 )
( 34 , 1.5521619715606843 )
( 35 , 1.5801608833097571 )
( 36 , 1.6078460069650617 )
( 37 , 1.6352320406020322 )
( 38 , 1.6623324871027236 )
( 39 , 1.6891597889157843 )
( 40 , 1.7157254437956162 )
( 41 , 1.7420401046966667 )
( 42 , 1.7681136663932193 )
( 43 , 1.793955340918174 )
( 44 , 1.8195737235362515 )
( 45 , 1.844976850665084 )
( 46 , 1.8701722509156156 )
( 47 , 1.8951669902273167 )
( 48 , 1.9199677119143885 )
( 49 , 1.944580672309712 )
( 50 , 1.9690117725855387 )
( 51 , 1.9932665872435074 )
( 52 , 2.017350389691476 )
( 53 , 2.0412681752657917 )
( 54 , 2.06502468200539 )
( 55 , 2.088624409442165 )
( 56 , 2.112071635635708 )
( 57 , 2.13537043265044 )
( 58 , 2.1585246806466216 )
( 59 , 2.1815380807355034 )
( 60 , 2.2044141667294923 )
( 61 , 2.2271563159018974 )
( 62 , 2.249767758858102 )
( 63 , 2.272251588605748 )
( 64 , 2.2946107689043345 )
( 65 , 2.316848141961689 )
( 66 , 2.3389664355413906 )
( 67 , 2.3609682695336787 )
( 68 , 2.382856162041321 )
( 69 , 2.4046325350218503 )
( 70 , 2.426299719526682 )
( 71 , 2.4478599605720923 )
( 72 , 2.469315421672789 )
( 73 , 2.4906681890677262 )
( 74 , 2.511920275661808 )
( 75 , 2.533073624709285 )
( 76 , 2.5541301132564733 )
( 77 , 2.5750915553657023 )
( 78 , 2.59595970513546 )
( 79 , 2.61673625953377 )
( 80 , 2.637422861057974 )
( 81 , 2.6580211002345138 )
( 82 , 2.678532517969902 )
( 83 , 2.6989586077644354 )
( 84 , 2.719300817798068 )
( 85 , 2.7395605528965383 )
( 86 , 2.7597391763881767 )
( 87 , 2.7798380118571515 )
( 88 , 2.7998583448002705 )
( 89 , 2.8198014241957345 )
( 90 , 2.8396684639859515 )
( 91 , 2.859460644484546 )
( 92 , 2.879179113708563 )
( 93 , 2.8988249886429713 )
( 94 , 2.9183993564409896 )
( 95 , 2.93790327556329 )
( 96 , 2.957337776862577 )
( 97 , 2.976703864612759 )
( 98 , 2.996002517490119 )
( 99 , 3.015234689507256 )
};
\addplot[
    domain=5:99,
    samples=200,
    dotted,
    green,
    very thick
]
{x^(1/1.5-1)};

\end{axis}

\end{tikzpicture}
    \caption{The optimality gap of the DROP (blue), SROP (red), and TA-DROP (black) algorithms for an initial inventory of $K=5$, with willingness-to-pay drawn from a Pareto distribution (shape parameter $\beta = 1.5$). The dotted green curve illustrates $n^{\frac{1}{\beta}-1}$, representing the theoretical asymptotic convergence limit approached by TA-DROP.}
    \label{fig:pareto}
\end{figure}

To conclude this section, it is instructive to ground these mathematical findings in their underlying economic intuition. The heavy-tailed nature of the Pareto distribution implies that the survival function $1-F(p) = p^{-\beta}$ decays at a slow, polynomial rate rather than an exponential one, meaning extreme valuations occur with significantly higher frequency. Economically, this dictates that demand remains remarkably resilient to price increases at the upper extremes. While artificially inflating prices under an exponential distribution would precipitate a precipitous drop in the probability of a sale, a proportional price increase under a Pareto distribution only marginally diminishes demand. It is precisely this structural resilience at the tail that makes deliberate price inflation not just profitable, but strictly necessary to achieve asymptotic optimality.

\section{Conclusion}\label{sec:conclusion}
In this paper, we have demonstrated that optimal revenue management in large markets does not require the deployment of complex, opaque optimization models. By anchoring simple, rule-based pricing heuristics to the structural properties of the willingness-to-pay distribution, we have shown that sellers can achieve near-optimal results while maintaining strict operational interpretability.

Our analysis provides both conceptual and technical contributions. Conceptually, we demonstrated that the tail behavior of customer demand dictates the fundamental strategy for inventory allocation: while simple run-out heuristics like DROP are sufficient for light-tailed (exponential) demand, heavy-tailed (Pareto) regimes necessitate a shift toward "tail-aware" pricing. By intentionally inflating prices to account for the option value of extreme valuations, we introduced the TA-DROP heuristic, which restores asymptotic optimality where standard methods fail. Technically, we established that these heuristics achieve a vanishing absolute optimality gap—decaying at a strict polynomial rate—providing a theoretical guarantee that is stronger than standard ratio-based convergence.

While our framework abstracts away the complexities of multi-night bookings and network effects, it serves as a robust and tractable foundation for practical pricing. We believe these findings offer a compelling case for the "simple-yet-optimal" paradigm in revenue management and open several interesting avenues for future research. On the theoretical front, extending this analysis beyond specific distribution classes to characterize generalized demand—such as identifying the precise hazard rate boundary that necessitates the TA-DROP heuristic—would significantly generalize our insights. Finally, it remains an open and highly relevant question whether the strict polynomial regret bounds achieved by these tools persist in the true fluid scaling regime, where the inventory capacity $K$ is exceptionally large and scales proportionally with the market.

\bibliographystyle{plainnat}
\bibliography{refs}

\begin{thebibliography}{13}
\providecommand{\natexlab}[1]{#1}
\providecommand{\url}[1]{\texttt{#1}}
\expandafter\ifx\csname urlstyle\endcsname\relax
  \providecommand{\doi}[1]{doi: #1}\else
  \providecommand{\doi}{doi: \begingroup \urlstyle{rm}\Url}\fi

\bibitem[Abdallah and Reed(2025)]{abdallah2025dynamic}
Tarek Abdallah and Josh Reed.
\newblock Dynamic pricing in the large market regime.
\newblock \emph{Available at SSRN 5223633}, 2025.

\bibitem[Abdallah and Reed(2026)]{abdallah2026regime}
Tarek Abdallah and Josh Reed.
\newblock Regime-dependent approximations for the single-item dynamic pricing problem.
\newblock \emph{Operations Research}, 74\penalty0 (3):\penalty0 1388--1407, 2026.

\bibitem[Bitran and Caldentey(2003)]{BC}
Gabriel Bitran and Ren{\'e} Caldentey.
\newblock An overview of pricing models for revenue management.
\newblock \emph{Manufacturing \& Service Operations Management}, 5\penalty0 (3):\penalty0 203--229, 2003.

\bibitem[Chen and Farias(2013)]{CF}
Yiwei Chen and Vivek~F Farias.
\newblock Simple policies for dynamic pricing with imperfect forecasts.
\newblock \emph{Operations Research}, 61\penalty0 (3):\penalty0 612--624, 2013.

\bibitem[Gallego and Topaloglu(2019)]{GT}
Guillermo Gallego and Huseyin Topaloglu.
\newblock \emph{Revenue management and pricing analytics}, volume 209.
\newblock Springer, 2019.

\bibitem[Gallego and Van~Ryzin(1994)]{Gvr}
Guillermo Gallego and Garrett Van~Ryzin.
\newblock Optimal dynamic pricing of inventories with stochastic demand over finite horizons.
\newblock \emph{Management science}, 40\penalty0 (8):\penalty0 999--1020, 1994.

\bibitem[Jasin(2014)]{Jasin}
Stefanus Jasin.
\newblock Reoptimization and self-adjusting price control for network revenue management.
\newblock \emph{Operations Research}, 62\penalty0 (5):\penalty0 1168--1178, 2014.

\bibitem[Maglaras and Meissner(2006)]{MM}
Constantinos Maglaras and Joern Meissner.
\newblock Dynamic pricing strategies for multiproduct revenue management problems.
\newblock \emph{Manufacturing \& Service Operations Management}, 8\penalty0 (2):\penalty0 136--148, 2006.

\bibitem[McAfee and te~Velde(2008)]{mcafee2008dynamic}
R~Preston McAfee and Vera te~Velde.
\newblock Dynamic pricing with constant demand elasticity.
\newblock \emph{Production and operations Management}, 17\penalty0 (4):\penalty0 432--438, 2008.

\bibitem[Reiman and Wang(2008)]{reiman2008asymptotically}
Martin~I Reiman and Qiong Wang.
\newblock An asymptotically optimal policy for a quantity-based network revenue management problem.
\newblock \emph{Mathematics of Operations Research}, 33\penalty0 (2):\penalty0 257--282, 2008.

\bibitem[R{\'e}nyi(1953)]{renyi1953theory}
Alfr{\'e}d R{\'e}nyi.
\newblock On the theory of order statistics.
\newblock 1953.

\bibitem[Wang and Wang(2022)]{wang2022constant}
Yining Wang and He~Wang.
\newblock Constant regret resolving heuristics for price-based revenue management.
\newblock \emph{Operations Research}, 70\penalty0 (6):\penalty0 3538--3557, 2022.

\bibitem[Yan(2011)]{Qiqi}
Qiqi Yan.
\newblock Mechanism design via correlation gap.
\newblock In \emph{Proceedings of the twenty-second annual ACM-SIAM symposium on Discrete Algorithms}, pages 710--719. SIAM, 2011.

\end{thebibliography}

\newpage
\begin{appendix}
\setcounter{page}{1}
\renewcommand{\thepage}{A\arabic{page}}
\section{Proofs}\label{sec:proofs}
Before Proving the propositions of the main body, we present a generalization of the Stirling approximation which will be useful in the future arguments.
\begin{restatable}{lemma}{StirlingApprox}
\label{StirlingApprox}For any positive integer $k$ and constants $c$, $M > k+c$, and $\alpha < 1$ we have:
    $$\lim_{n \rightarrow \infty} \left[\sum_{i=M}^n \left(\frac{{i-c \choose k-1}}{{n \choose k}} \; \log (i)\right) - \log(n) +  \frac{1}{k}\right]n^{\alpha} =0$$
\end{restatable}
\begin{remark}
    As sanity check, for $k=1$ this is the same as Stirling approximation:
    $$\frac{\log(n!)}{n} \approx \log(n) - 1$$
\end{remark}
\begin{proof}
    We use integral approximation:
    \begin{align*}
    \frac{1}{{n \choose k}(k-1)!}\int_{M}^{n+1} (x-c)...(x-c-k+2) \log(x) \; dx - \log(n) + \frac{1}{k} &\\
    \geq \sum_{i=M}^n \frac{{i-c \choose k-1}}{{n \choose k}} \; \log (i) - \log(n) + \frac{1}{k} & \geq \\
     \frac{1}{{n \choose k}(k-1)!}\int_{M-1}^{n} (x-c)...(x-c-k+2) \log(x)\; dx - \log(n) + \frac{1}{k} &
    \end{align*}
    We demonstrate that the upper bound time $n^{\alpha}$ converges to $0$, and the argument for the lower bound follows the same steps exactly. We have:
    \begin{align*}
    & \lim_{n \rightarrow \infty} \left[\frac{n^{\alpha}}{{n \choose k}(k-1)!}\int_{M}^{n+1} (x-c)...(x-c-k+2) \log(x) dx - n^{\alpha} \log (n) + \frac{n^{\alpha}}{k}\right]\\
    =& \lim_{n \rightarrow \infty} \left[ \frac{k\int_{M}^{n+1} (x-c)...(x-c-k+2) \log(x)dx}{n^{k-\alpha}}  - n^{\alpha} \log (n) + \frac{n^{\alpha}}{k}\right]
    \\ =&\lim_{n \rightarrow \infty} \left[ \frac{k\int_{M}^{n+1} x^{k-1} \log(x) dx}{n^k} - n^{\alpha} \log (n) + \frac{n^{\alpha}}{k}\right]\\
    =& \lim_{n \rightarrow \infty}  \left(\left(\frac{x^k\log(x)}{n^{k-\alpha}}-\frac{x^k}{n^{k-\alpha}k}\middle)\right|_M^{n+1} - n^{\alpha} \log (n) + \frac{n^{\alpha}}{k}\right) = 0 
    \end{align*}
    Similarly, we can demonstrate the same convergence for the lower bound, and the result follows from the Squeeze Theorem. $\Box$
\end{proof}

\subsection{Proofs of \Cref{sec:model}}

\propSROP*
\begin{proof}
    The probability of selling to each customer is given by $\frac{k}{n}$. Therefore, the probability of selling to $i$ customers can be expressed as:
    \begin{align*}
        \prob(i \text{ Sales}) = \lim_{n \rightarrow \infty} {n \choose i}\left(1-\frac{k}{n}\right)^{n-i} \left(\frac{k}{n}\right)^i = \frac{n^i}{i!}\cdot e^{-k}\cdot \frac{k^i}{n^i} = e^{-k}\frac{k^i}{i!}
    \end{align*}
    Hence, the expected number of sales is:
    \begin{align*}
        \EX [\text{Sales}] &= \sum_{i=0}^{k-1} i \cdot \prob(i \text{ Sales}) + k \cdot \left(1- \sum_{i=0}^{k-1} \prob(i \text{ Sales}) \right)\\
        &= k - \sum_{i=0}^{k-1} (k-i) \prob(i \text{ Sales})\\
        &= k-e^{-k}\left(\sum_{i=0}^{k-1}\frac{k^{i+1}}{i!} - \sum_{i=1}^{k-1}\frac{k^{i}}{(i-1)!}\right) = k-e^{-k}\frac{k^k}{(k-1)!}
    \end{align*}
\end{proof}

\uniformD*
\begin{proof}
    We prove that $\lim_{\customer \rightarrow \infty} D(\inventory,\customer) = \bar{p}\inventory$ using induction on $\inventory$. The base case is straightforward for $\inventory = 0$. Now, assume the statement holds for  $\inventory-1$, and we aim to prove it for $\inventory$. Then for any $\epsilon > 0$, there exists a sufficiently large integer $M_{\epsilon}$ such that for all $\customer \geq M_{\epsilon}$, we have: $$ x_k(n) \coloneqq D(\inventory-1,\customer-1)+ \bar{q}_F\left(\frac{\inventory}{\customer}\right) > k\bar{p}-\epsilon$$
    Notice that at state $(k,n)$, $x_k(n)$ represents the revenue of the algorithm, conditional on a purchase being made. Using this lower bound, \eqref{DICRec} can be rewritten for large $n >> M_{\epsilon}$:
    \begin{align*}
    D(\inventory,n) =& \left(1-\frac{\inventory}{n}\right)D(\inventory,n-1) + \frac{\inventory}{n} x_k(n)\\
    =& \frac{n-\inventory}{n}\left(\frac{n-1-\inventory}{n-1}D(\inventory,n-2) + \frac{\inventory}{n-1} x(n-1)\right)+ \frac{\inventory}{n} x_k(n)\\
    >& \sum_{i=M_{\epsilon}}^n \frac{(n-k)!}{(i-k)!}\cdot \frac{i!}{n!}\cdot \frac{k}{i} \; x(i)\\
    >& \sum_{i=M_{\epsilon}}^n \frac{(n-k)!}{(i-k)!}\cdot \frac{i!}{n!}\cdot \frac{k}{i} \; (k\bar{p}-\epsilon)\\
    =& (k\bar{p}-\epsilon)\sum_{i=M_{\epsilon}}^n \frac{{i-1 \choose k-1}}{{n \choose k}} = (k\bar{p}-\epsilon) \left( 1-\sum_{i=k}^{M_{\epsilon}-1} \frac{{i-1 \choose k-1}}{{n \choose k}}\right) 
\end{align*}
Where the last equation comes from \textit{Hockey Stick Identity}\footnote{Hockey Stick Identity: $\sum_{i=k}^n {i-1 \choose k-1} = {n \choose k}$. This can be proven by induction on $n$ and using Pascal's triangle formula.}. Then by making the number of customers $n$ larger we can get:
\begin{align*}
    \lim_{n \rightarrow \infty} D(\inventory,n) \geq \lim_{n \rightarrow \infty} (k\bar{p}-\epsilon) \left( 1-\sum_{i=k}^{M_{\epsilon}-1} \frac{{i-1 \choose k-1}}{{n \choose k}}\right) = k\bar{p}-\epsilon
\end{align*}
This shows that the limit of the DROP can be made arbitrarily close to $\bar{p}k$, thereby completing the proof. $\Box$
\end{proof}

\subsection{Proofs of \Cref{sec:mainResults}}

\propD*
\begin{proof}
    similar to the proof of the \Cref{propUniform} we proceed by using induction on $k$ and expand the recursion as done in that proof:
    \begin{align*}
    D(\inventory,n) = &
    \left(1-\frac{\inventory}{n} \right)
    D(\inventory,n-1) + \frac{\inventory}{n}\left( D(\inventory-1,n-1) + \bar{q}_F\left(\frac{\inventory}{n}\right)\right)\\
    =& \left(1-\frac{\inventory}{n} \right)
    D(\inventory,n-1) + \frac{\inventory}{n}\left( D(\inventory-1,n-1) + \log \left(\frac{n}{\inventory}\right)\right)\\
    =& \sum_{i=k+1}^n \frac{{i-1 \choose k-1}}{{n \choose k}} \; \left(D(\inventory-1,i-1) + \log \left(\frac{i}{\inventory}\right) \right)
    \end{align*}
    By defining $M_{\epsilon}$ in which for $n > M_{\epsilon}$ we have 
    $$\left| (k-1)\log(n-1) - D(k-1,n-1) - k + 1 - \log ((k-1)!) \right| \leq  \epsilon n^{-\alpha}$$
    it follows that:
    \begin{align*}
    & \sum_{i=M_{\epsilon}+1}^n \frac{{i-1 \choose k-1}}{{n \choose k}} \; \left(k\log (i-1) - k + 1 - \log(k!) - \epsilon (i-1)^{-\alpha} \right)\\
    \leq &  D(\inventory,n) =  \sum_{i=k+1}^n \frac{{i-1 \choose k-1}}{{n \choose k}} \; \left(D(\inventory-1,i-1) + \log \left(\frac{i}{\inventory}\right) \right)\\
    \leq & \sum_{i=M_{\epsilon}+1}^n \frac{{i-1 \choose k-1}}{{n \choose k}} \; \left(k\log (i) - k + 1 - \log(k!) + \epsilon (i-1)^{-\alpha} \right)
    + \sum_{i=k+1}^{M_{\epsilon}} \frac{{i-1 \choose k-1}}{{n \choose k}} \; C
    \end{align*}
    for a constant $C$. Then we only need to prove that both upper and lower bounds minus $k\log(n)- k - \log(k!)$ times $n^{\alpha}$ are converging to $0$.\\
    \textbf{Lower bound.} We utilize \Cref{StirlingApprox} and Hockey Stick Identity to simplify the summations:
    \begin{align*}
    &\lim_{n \rightarrow \infty} \left[k\log(n) - k - \log(k!) - \sum_{i=M_{\epsilon}+1}^n \frac{{i-1 \choose k-1}}{{n \choose k}} \; \left(k\log (i-1) - k + 1 - \log(k!) - \epsilon (i-1)^{-\alpha} \right)\right]n^{\alpha} \\
    =& k\lim_{n \rightarrow \infty} \left[\log(n) - \sum_{i=M_{\epsilon}+1}^n \frac{{i-1 \choose k-1}}{{n \choose k}} \; \log (i-1) - \frac{1}{k}\right]n^{\alpha} + \epsilon
    = \epsilon 
    \end{align*}
    \textbf{Upper bound.} Notice that the constant term converges to zero as $n$ approaches to infinity and  We can apply the same approach used for the lower bound to demonstrate the convergence of the upper bound.
    
    Putting all the pieces together we showed that for all $\epsilon>0$ there is a sufficiently large $M_\epsilon$ such that for all $n>M_\epsilon$:
    \begin{align*}
        -\epsilon \leq [k\log(n) -k-\log(k!) - D(k,n)]n^{\alpha} \leq \epsilon \;\;\;\Box
    \end{align*}
\end{proof}

\propV*
\begin{proof} First notice that with an exponential willingness to pay, the recursion will be of the following form:
    \begin{align*}
        V(k,n) = \max_r e^{-r}\left(r + V(k-1,n-1)\right) + \left(1-e^{-r}\right)V(k,n-1)
    \end{align*}
    Solving for the optimal $r$ we get:
    \begin{align}\label{rStar}
        -e^{-r^*}\left(r^* + V(k-1,n-1)\right) + e^{-r^*}V(k,n-1) + e^{-r^*} =0 \nonumber\\
        \iff  r^* = V(k,n-1) + 1 - V(k-1,n-1)
    \end{align}
    Hence, 
    \begin{align*}
        V(k,n+1) = V(k,n) + e^{-V(k,n) - 1 + V(k-1,n)}
    \end{align*}
    Let $\Tilde{V}(\inventory,\customer) \coloneqq  V(\inventory,\customer) - \left(\inventory\log(\customer)-\inventory - \log(\inventory !)\right)$. Then:
    \begin{align*}
        \Tilde{V}(\inventory,\customer+1)=  \Tilde{V}(\inventory,\customer) + \frac{k}{n} \left(e^{\Tilde{V}(\inventory-1,\customer)-\Tilde{V}(\inventory,\customer)} - \customer \log \left(\frac{\customer+1}{\customer}\right)\right)
    \end{align*}
    We just need to show $\lim_{\customer \rightarrow \infty} n^{\alpha}\Tilde{V}(\inventory,\customer) = 0$; we do this by induction on $\inventory$. For $\inventory = 0$ it is trivial as $\tilde{V}(\inventory,\customer) = 0$. Assume the statement is correct for $k-1$. Define
    $$W(k,n) = \tilde{V}(k,n) - \tilde{V}(k-1,n)$$
    Using induction, if we show $\lim_{\customer \rightarrow \infty} n^{\alpha}W(\inventory,\customer) = 0$ converges to zero we are done. The recursion becomes:
    \begin{align*}
        W(k,n+1)= W(k,n) + \frac{1}{n}\left( ke^{-W(k,n)} - (k-1) e^{-W(k-1,n)} - n\log \left(1+\frac{1}{n}\right)\right)
    \end{align*}
Notice the initial values for $W$ are $W(1,1) = 1+e^{-1}$ and $W(k,1) = 1+\log(k)$ for $k > 1$. The rest of the proof has several steps that we break them into lemmas.
    \begin{lemma}\label{nlogn}
        For $n \geq 1$ we have:
        \begin{align*}
            1 \geq n\log\left(1+\frac{1}{n}\right) \geq 1- \frac{1}{2n}
        \end{align*}
    \end{lemma}
    \begin{proof} 
    Notice for $n \geq 1$ the Taylor expansion of $\log\left(1+\frac{1}{n}\right)$ can be written as
    \begin{align*}
        \log\left(1+\frac{1}{n}\right) = \frac{1}{n} -\frac{1}{2n^2} + \frac{1}{3n^3} - \dots
    \end{align*}
    In order to show $\log\left(1+\frac{1}{n}\right) \leq \frac{1}{n}$ we use AM-GM inequality to show:
    \begin{align*}
        \frac{1}{2n^2} + \frac{1}{4n^4} + \frac{1}{6n^6} + \dots \geq& \frac{1}{\sqrt{8}n^3} + \frac{1}{\sqrt{24}n^5} + \frac{1}{\sqrt{48}n^7} + \dots\\
        \geq& \frac{1}{3n^3} + \frac{1}{5n^5} + \frac{1}{7n^7} + \dots
    \end{align*}
    To show $\log\left(1+\frac{1}{n}\right) \geq \frac{1}{n} - \frac{1}{2n^2}$ we do similarly to show:
    \begin{align*}
        \frac{1}{3n^3} + \frac{1}{5n^5} + \frac{1}{7n^7} + \dots \geq& \frac{1}{\sqrt{15}n^4} + \frac{1}{\sqrt{35}n^6} + \frac{1}{\sqrt{63}n^8} + \dots\\
        \geq& \frac{1}{4n^4} + \frac{1}{6n^6} + \frac{1}{8n^8} + \dots \qed
    \end{align*}
    \end{proof}
\begin{lemma}
    $W(k,n)$ is always positive.
\end{lemma}
\begin{proof}
Using the \cref{nlogn} and also by strong induction we can prove $W(k,n)$ is always positive. Note that
\begin{align*}
    W(k,n+1) \geq W(k,n) + \frac{1}{n}\left( k\left(1-W(k,n)\right) - (k-1)-1\right) = \frac{n-k}{n}W(k,n).
\end{align*}
Therefore, we are only left to show it is positive for $k > n$. In the following by induction we show for $k>n$ we have:
\begin{align*}
    W(k,n) = 1+\log(k)-\log(n)
\end{align*}
This is easy to see by using induction on $n$. For $n=1$ it follows from the definition. Now assume $k> n+1$, then $k > n$ and:
\begin{align*}
    W(k,n+1)=& W(k,n) + \frac{1}{n}\left( ke^{-W(k,n)} - (k-1) e^{-W(k-1,n)} - n\log \left(1+\frac{1}{n}\right)\right)\\
    =& 1+\log(k)-\log(n) + \frac{1}{n}\left(- n\log \left(1+\frac{1}{n}\right)\right) = 1+\log(k)-\log(n+1)
\end{align*}
Therefore $W(k,n)$ is always positive. $\qed$
\end{proof} 
\begin{lemma}
    $W(k,n)$ converges to zero.
\end{lemma}
\begin{proof}
    Our proof uses a contradiction approach together with induction on $k$. Assume there is a positive number $\epsilon$ where $W(k,n) > -\log(1-\epsilon)$ infinitely many times. Using the induction and \Cref{nlogn} we know there exists a large $N_{\epsilon}$ such that for all $n > N_{\epsilon}$, $W(k-1,n) < -\log\left(1-\frac{\epsilon}{2}\right)$ and $n\log\left(1+\frac{1}{n}\right)> 1-\frac{\epsilon}{2}$.\\
    Now suppose for all $n$ we have $W(k,n) > -\log(1-\epsilon)$. Therefore we can write:
    \begin{align*}
        W(k,n+1)=& W(k,n) + \frac{1}{n}\left( ke^{-W(k,n)} - (k-1) e^{-W(k-1,n)} - n\log \left(1+\frac{1}{n}\right)\right)\\
        \leq& W(k,n) + \frac{1}{n}\left( k(1-\epsilon) - (k-1) (1-\frac{\epsilon}{2}) - (1-\frac{\epsilon}{2})\right) = W(k,n) - \frac{k\epsilon}{2n}
    \end{align*}
    The fact that sum of reciprocals are diverging gives a contradiction. Therefore, the only remaining possible case is that $W(k,n)$ oscillates and passes $-\log(1-\epsilon)$ infinitely many times. Take one of this instances and assume for $n>N$, $W(k,n+1) \geq -\log(1-\epsilon) \geq W(k,n)$. The fact that $W(k,n+1) \geq W(k,n)$ means that $W(k,n) \leq -\log \left(1-\frac{\epsilon}{2}\right)$. Therefore:
    \begin{align*}
        W(k,n+1)=& W(k,n) + \frac{1}{n}\left( ke^{-W(k,n)} - (k-1) e^{-W(k-1,n)} - n\log \left(1+\frac{1}{n}\right)\right)\\
        \leq& -\log \left(1-\frac{\epsilon}{2}\right) + \frac{1}{n}\left( k - (k-1) (1-\frac{\epsilon}{2}) - (1-\frac{\epsilon}{2})\right) = -\log \left(1-\frac{\epsilon}{2}\right) + \frac{k\epsilon}{2n} 
    \end{align*}
    But for large enough $n$ the right hand side is less than $-\log(1-\epsilon)$ which completes our proof.$\qed$
\end{proof}
\begin{lemma} For all $\alpha < 1$, 
    $\lim_{n\rightarrow \infty} W(k,n) n^{\alpha} = 0$
\end{lemma}
    Since $W(k,n)$ converges to zero, for each $\beta < 1$ (we will later set $\beta$ to be greater than $\alpha$) there is a large number $N_{\beta}$ where for all $n>N_{\beta}$ we have $e^{-W(k,n)} \leq 1-\beta W(k,n)$. For such $n$ we have:
    \begin{align*}
        W(k,n+1) &\leq W(k,n) + \frac{1}{n}\left( k\left(1-\beta W(k,n)\right) - (k-1)\left(1-W(k-1,n)\right) -1+\frac{1}{2n}\right) \\
        &= \frac{n-k\beta}{n}W(k,n) + \frac{k-1}{n}W(k-1,n)+\frac{1}{2n^2}\\
        &= W(k,N_{\beta})\left(\prod_{i = N_{\beta}}^n \frac{i-k\beta}{i}\right) + (k-1) \sum_{\ell=N_{\beta}}^n \left( \frac{W(k-1,\ell)}{\ell} \prod_{i = \ell+1}^n \frac{i-k\beta}{i}\right) + \sum_{\ell=N_{\beta}}^n \frac{1}{2\ell^2} \prod_{i = \ell+1}^n \frac{i-k\beta}{i}.
        \\
        &\overset{(a)}{\leq} (C+ \epsilon_n)\left(W(k,N_{\beta}) \left(\frac{N_{\beta}}{n}\right)^{k\beta} + (k-1) \sum_{\ell=N_{\beta}}^n \left( \frac{W(k-1,\ell)}{\ell} 
        \left(\frac{\ell+1}{n}\right)^{k \beta}
        \right) + \sum_{\ell=N_{\beta}}^n \frac{1}{2\ell^2} \left(\frac{\ell+1}{n}\right)^{k \beta}\right).
    \end{align*}
Which (a) comes from the \Cref{lem:ndividedM} and also \Cref{remark:corr}.
Therefore taking the limit and multiplying both sides by $\alpha$ and setting $\beta > \alpha$ we get:
\begin{align*}
        \lim_{n\rightarrow \infty} \frac{n^{\alpha} W(k,n+1)}{C+ \epsilon_n}
        \leq \lim_{n\rightarrow \infty} W(k,N_{\beta}) \frac{N_{\beta}^{k\beta}}{n^{k\beta-\alpha}} + (k-1) \sum_{\ell=N_{\beta}}^n \left( \frac{\ell^{\alpha}W(k-1,\ell)}{\ell^{1+\alpha}} 
        \frac{(\ell+1)^{k \beta}}{n^{k \beta-\alpha}}\right) + \sum_{\ell=N_{\beta}}^n \frac{1}{2\ell^2} \frac{(\ell+1)^{k\beta}}{n^{k \beta-\alpha}}.
    \end{align*}
Notice that the first term goes to zero as $k\beta \geq \beta > \alpha$. Using induction we can simplify the second term as we know $\lim_{n\rightarrow\infty} n^{\alpha} W(k-1,n) = 0$. By induction we know for all $n>\hat{N}_{\beta}$, $n^{\alpha} W(k-1,n) < \epsilon$. The third term is the sum of $\frac{\ell^{k\beta-2}}{n^{k\beta-\alpha}}$ which will be of the order $\frac{n^{k\beta - 1}}{n^{k\beta - \alpha}} \rightarrow 0$. Therefore, the only remaining term is the second term:
\begin{align*}
        \lim_{n\rightarrow \infty} \frac{n^{\alpha} W(k,n+1)}{C+ \epsilon_n}
        \leq \lim_{n\rightarrow \infty}  (k-1) \epsilon \sum_{\ell=\hat{N}_{\beta}}^n \left( \frac{1}{\ell^{1+\alpha}} 
        \frac{(\ell+1)^{k \beta}}{n^{k \beta-\alpha}}\right) = \epsilon(k-1) \Theta.
    \end{align*}
where $\Theta$ is a constant because sum of $l^{k\beta-1-\alpha}$ will be of the order $n^{k\beta-\alpha}$. Therefore, we can decrease the right hand side arbitrary by reducing $\epsilon$ and hence this shows that $\lim_{n\rightarrow \infty} n^{\alpha} W(k,n+1)=0 \qed$
\begin{restatable}{lemma}{GammaFunction}\label{lem:ndividedM}
    for any $M > \alpha$:
    \footnote{
    In fact $C = 1$, but our analysis does not rely on this exact value.}
    \begin{align*}
        \lim_{N \rightarrow \infty} \left(\frac{N}{M}\right)^\alpha \prod_{n=M}^N \frac{n-\alpha}{n} = C
    \end{align*}
\end{restatable}
\begin{remark}\label{remark:corr}
    The above result remains valid even when \( N - M \) is fixed and both \( N \) and \( M \) grow large.
\end{remark}
\end{proof}
\begin{proof}
    We take the logarithm of the left-hand side and apply the Taylor expansion $\log(1-x) = -x - \frac{x^2}{2} - ...$ we get:
    \begin{align*}
        \log \left(\left(\frac{N}{M}\right)^\alpha \prod_{n=M}^N \frac{n-\alpha}{n}\right) &= \alpha\log \left(N\right) - \alpha\log \left(M\right) + \sum_{n=M}^N\log \left(1-\frac{\alpha}{n}\right)\\
        &= \alpha\log \left(N\right) - \alpha\log \left(M\right) + O(1) - \sum_{n=M}^N \frac{\alpha}{n}
    \end{align*}
    Using the fact that 
\[
\sum_{n=M}^N \frac{1}{n} = \log N - \log M + O(1),
\]
we conclude that the entire expression is \(O(1)\). Exponentiating both sides yields the desired result. $\qed$
\end{proof}
\begin{remark}\label{remark:optimalPrices}
    Looking at \Cref{rStar} we can see:
    \begin{align*}
        r^* &= V(k,n-1) + 1 - V(k-1,n-1)\\
        &\approx k\log(n-1)  -k - \log(k!) + 1 \\
        &- (k-1)\log(n-1)  - (k-1) - \log((k-1)!)\\
        &= \log(n-1)-\log(k) \approx \log(\frac{n}{k})
    \end{align*}
    Then the prices posted by DROP is also converging to the optimal prices.
\end{remark}
\subsection{Proofs of \Cref{sec:pareto}}
\propParetoD*
\begin{proof}
We know
    \begin{align*}
        D(k,n) &= D(k,n-1) + p^{-\beta}\left(D(k-1,n-1)-D(k,n-1) + p\right)\\ &= D(k,n-1) + \frac{C_k^{\frac{\beta}{\beta-1}}}{n}\left(D(k-1,n-1)-D(k,n-1) + C_k^{\frac{1}{1-\beta}}n^{\frac{1}{\beta}}\right)
    \end{align*}    
we prove this by induction. Assume the statement is correct for $k-1$ (with the base being $k=0$) and we want to prove it for $k$. Define $\epsilon_n = D(k,n) - C_kn^{\frac{1}{\beta}}$ and $\delta_n = D(k-1,n) - C_{k-1}n^{\frac{1}{\beta}}$. By induction we know $\lim_{n\rightarrow\infty} \delta_nn^{\alpha} = 0$ and want to show the same for $\epsilon$. Rewriting the recursion using $\epsilon,\delta$ we will have:
\begin{align*}
\epsilon_n  =&  - C_kn^{\frac{1}{\beta}}+ \epsilon_{n-1}+ C_k(n-1)^{\frac{1}{\beta}} + \frac{C_k^{\frac{\beta}{\beta-1}}}{n}\left(\delta_{n-1} + C_{k-1}(n-1)^{\frac{1}{\beta}} - \epsilon_{n-1} - C_k(n-1)^{\frac{1}{\beta}} + C_k^{\frac{1}{1-\beta}}n^{\frac{1}{\beta}}\right)\\
=& - C_kn^{\frac{1}{\beta}}+ \epsilon_{n-1}+ C_k(n-1)^{\frac{1}{\beta}} + \frac{C_k^{\frac{\beta}{\beta-1}}}{n}\left(\delta_{n-1} - \epsilon_{n-1} - \frac{\beta-1}{\beta}C_k^{\frac{1}{1-\beta}}(n-1)^{\frac{1}{\beta}} + C_k^{\frac{1}{1-\beta}}n^{\frac{1}{\beta}}\right)\\
=&  \epsilon_{n-1} + \frac{C_k^{\frac{\beta}{\beta-1}}}{n}\left(\delta_{n-1} - \epsilon_{n-1}\right)+C_k\left(- n^{\frac{1}{\beta}}+(n-1)^{\frac{1}{\beta}}+\frac{n^{\frac{1}{\beta}}-(n-1)^{\frac{1}{\beta}}}{n}+\frac{1}{\beta}\cdot\frac{(n-1)^{\frac{1}{\beta}}}{n}\right) \\
=& \epsilon_{n-1} + \frac{C_k^{\frac{\beta}{\beta-1}}}{n}\left(\delta_{n-1} - \epsilon_{n-1}\right)+\mathcal{O}(n^{\frac{1}{\beta}-2})\\
=& \left(1-\frac{C_k^{\frac{\beta}{\beta-1}}}{n}\right)\epsilon_{n-1} + \mathcal{O}(n^{\frac{1}{\beta}-2+\zeta})
\end{align*}
Where the last equation came from the fact that $\delta_n = \mathcal{O}(n^{\frac{1}{\beta}-1+\zeta})$ for every $\zeta > 0$.
Using triangle inequality, this shows that:
\begin{align*}
    |\epsilon_n| \leq& \sum_{i=1}^n \mathcal{O}(i^{\frac{1}{\beta}-2+\zeta})\prod_{j=i+1}^n \left(1-\frac{C_k^{\frac{\beta}{\beta-1}}}{j}\right) \\
    =& \sum_{i=1}^n \mathcal{O}(i^{\frac{1}{\beta}-2+\zeta}) \Theta \left((\frac{i}{n})^{C_k^{\frac{\beta}{\beta-1}}}\right)\\
    =&  (\frac{1}{n})^{C_k^{\frac{\beta}{\beta-1}}}\sum_{i=1}^n \mathcal{O}(i^{\frac{1}{\beta}-2+\zeta+C_k^{\frac{\beta}{\beta-1}}})\\
    =& (\frac{1}{n})^{C_k^{\frac{\beta}{\beta-1}}} \mathcal{O}(n^{\frac{1}{\beta}-1+\zeta+C_k^{\frac{\beta}{\beta-1}}})= \mathcal{O}(n^{\frac{1}{\beta}-1+\zeta})
\end{align*}
Therefore, for any $\alpha < 1-\frac{1}{\beta}$, $\lim_{n\rightarrow \infty} n^\alpha \epsilon_n = 0$.$\Box$
\end{proof}

Similar to this proof we can prove a more general lemma that will be useful in Proof of \Cref{propParetoV}.

\medskip
\begin{lemma}\label{lem:Lower}Let $\mu>\tilde{\alpha}$, let $(a_n)$ satisfy
$a_n\ge0$ and $a_n\le\left(1-\frac{\mu}{n-1}\right)a_{n-1}+s_n$ for all large $n$, with
$s_n=o\!\left(n^{-\tilde{\alpha}-1}\right)$. Then $a_n=o\!\left(n^{-\tilde{\alpha}}\right)$.
\end{lemma}

\begin{proof}
    \begin{align*}
    a_{n+1} \leq& \sum_{i=N+1}^n o(i^{-\tilde{\alpha}-1})\prod_{j=i}^{n-1} \left(1-\frac{\mu}{j}\right) + \mathcal{O}\left(\left(\frac{N}{n}\right)^{\mu}\right) \\
    =& \sum_{i=N+1}^n o(i^{-\tilde{\alpha}-1}) \Theta \left((\frac{i}{n})^{\mu}\right)\\
    =&  (\frac{1}{n})^{\mu}\sum_{i=N+1}^n o(i^{\mu-\tilde{\alpha}-1})\\
    =& (\frac{1}{n})^{\mu} o(n^{\mu-\tilde{\alpha}})= o(n^{-\tilde{\alpha}}). \hfill \Box
\end{align*}
\end{proof}

\propParetoV*
\begin{proof}
    Recall
    \begin{align*}
        V(k,n) = V(k,n-1) + \frac{\left(\beta - 1\right)^{\beta-1}}{\beta^\beta}(V(k,n-1)- V(k-1,n-1))^{1-\beta}
    \end{align*}
    Let us define $V(k,n) = C_kn^{\frac{1}{\beta}} + \epsilon_n$ and also $V(k-1,n) = C_{k-1}n^{\frac{1}{\beta}} + \delta_n$. By induction we know $\lim_{n \rightarrow \infty} n^{\tilde{\alpha}}\delta_n = 0$ for all $\tilde{\alpha}\in[\alpha,1-\frac{1}{\beta})$ and we want to show $\lim_{n \rightarrow \infty} n^{\alpha}\epsilon_n = 0$. Furthermore, define $d_n\coloneqq \epsilon_n-\delta_n$.

    Using the new notation the above recursion and the fact that $C_k^{\frac{1}{1-\beta}} = \frac{\beta}{\beta-1}\left(C_k - C_{k-1}\right)$:
    \begin{align}\label{eq:epsilon}
        \epsilon_n =& \epsilon_{n-1} + C_k\left((n-1)^{\frac{1}{\beta}}-n^{\frac{1}{\beta}}\right) + \frac{1}{\beta} \left(\frac{\beta}{\beta-1}\left(C_k-C_{k-1}\right)(n-1)^{\frac{1}{\beta}}+\frac{\beta}{\beta-1}d_{n-1}\right)^{1-\beta} \nonumber\\
        =& \epsilon_{n-1} + C_k\left((n-1)^{\frac{1}{\beta}}-n^{\frac{1}{\beta}}\right)  + \frac{1}{\beta} \left(C_k^{\frac{1}{1-\beta}}(n-1)^{\frac{1}{\beta}}+\frac{\beta}{\beta-1}d_{n-1}\right)^{1-\beta} \nonumber\\
        =& \epsilon_{n-1} + C_k\left((n-1)^{\frac{1}{\beta}}-n^{\frac{1}{\beta}} + \frac{1}{\beta} (n-1)^{\frac{1}{\beta}-1}\right)+ \frac{1}{\beta} \left(C_k^{\frac{1}{1-\beta}}(n-1)^{\frac{1}{\beta}}+\frac{\beta}{\beta-1}d_{n-1}\right)^{1-\beta} -\frac{C_k}{\beta} (n-1)^{\frac{1}{\beta}-1} \nonumber\\
        =& \epsilon_{n-1} + \frac{C_k}{2\beta}\left(1-\frac{1}{\beta}\right) (n-\zeta_n)^{\frac{1}{\beta}-2}+\frac{1}{\beta} \left(C_k^{\frac{1}{1-\beta}}(n-1)^{\frac{1}{\beta}}+\frac{\beta}{\beta-1}d_{n-1}\right)^{1-\beta} -\frac{C_k}{\beta} (n-1)^{\frac{1}{\beta}-1}
    \end{align}
    Doing the Taylor's expansion for the last two terms we will have:
    \begin{align*}
         \epsilon_{n-1} + \frac{C_k}{2\beta}\left(1-\frac{1}{\beta}\right) (n-1)^{\frac{1}{\beta}-2} -d_{n-1}
         C_k^{\frac{\beta}{\beta-1}}(n-1)^{-1} +\frac{\beta^2}{2(\beta-1)}d_{n-1}^2 C_k^{\frac{\beta+1}{\beta-1}}(n-1)^{-1-\frac{1}{\beta}} \geq \\
         \epsilon_n \geq \epsilon_{n-1} + \frac{C_k}{2\beta}\left(1-\frac{1}{\beta}\right) n^{\frac{1}{\beta}-2} -d_{n-1}
         C_k^{\frac{\beta}{\beta-1}}(n-1)^{-1}
    \end{align*}
    Note that the upper bound is valid for $d_{n-1} \geq 0$.
    
    \medskip
    
    \noindent\textbf{Step 1. Lower bound.} Assume $\epsilon_{n-1} <0$, Using the lower bound above for $n-1 > C_k^\frac{\beta}{\beta-1}$:
    \begin{align*}
        \epsilon_n \geq \epsilon_{n-1}\left(1-C_k^{\frac{\beta}{\beta-1}}(n-1)^{-1}\right) + \mathcal{O}\left(n^{\frac{1}{\beta}-2}\right)+\delta_{n-1}
         C_k^{\frac{\beta}{\beta-1}}(n-1)^{-1}
    \end{align*}
    We know $\delta_n = o(n^{-\tilde{\alpha}})$ for every $\tilde{\alpha}<1-\frac{1}{\beta}$. Pick one $\tilde{\alpha} > \max(0,\alpha)$. 
    This means that:
    \begin{align*}
        (-\epsilon_n)^+ \leq (-\epsilon_{n-1})^+\left(1-C_k^{\frac{\beta}{\beta-1}}(n-1)^{-1}\right) + |o(n^{-\tilde{\alpha}-1})|.
    \end{align*}
    where $(x)^+ \coloneqq \max(0,x)$. This is true obviously for $\epsilon_n > 0$. For the other case if $\epsilon_{n-1}$ is negative, it follows from multiplying the previous inequality with minus 1. If it's positive, then it follows from the fact that $\epsilon_n \geq \epsilon_{n-1}\left(1-C_k^{\frac{\beta}{\beta-1}}(n-1)^{-1}\right) + o(n^{-\tilde{\alpha}-1}) \geq  o(n^{-\tilde{\alpha}-1})$. Now, invoking \Cref{lem:Lower} (with $C_k^{\frac{\beta}{\beta-1}} \geq C_1^{\frac{\beta}{\beta-1}} = 1-\frac{1}{\beta} > \tilde{\alpha}$) we get that $(-\epsilon_n)^+ = o(n^{-\tilde{\alpha}})$. This means $\epsilon_n$ is lower bounded by $-o(n^{-\tilde{\alpha}})$. More formally, it shows that for every $C>0$ and $\tilde{\alpha} < 1-\frac{1}{\beta}$, there exist a large number $N$ that $\epsilon_n > -Cn^{-\tilde{\alpha}}$.
    \medskip
    
    \noindent\textbf{Step 2. Uniform upper bound.} In the case of $d_{n-1} > 0$, another useful upper bound for $\epsilon_n$ is to lower bound $d_{n-1} = 0$ as the third minus fourth term in \Cref{eq:epsilon} is negative for positive $d_{n-1}$. Hence:
    \begin{align*}
        \epsilon_n \leq \epsilon_{n-1} +  \frac{C_k}{2\beta}\left(1-\frac{1}{\beta}\right) (n-1)^{\frac{1}{\beta}-2}
    \end{align*}
    which is summable as $\frac{1}{\beta}-2 < -1$. Therefore as $\delta_n$ is converging to zero, the only possibility for an unbounded $\epsilon_n$ is that it jumps when $d_{n-1} < 0$ but from step 1 we know $d_n \geq -o(n^{-\tilde{\alpha}})$:
    \begin{align*}
        \epsilon_n \leq \epsilon_{n-1} + \mathcal{O}(n^{\frac{1}{\beta}-2})+\frac{1}{\beta} \left(C_k^{\frac{1}{1-\beta}}(n-1)^{\frac{1}{\beta}}-o(n^{-\tilde{\alpha}})\right)^{1-\beta} -\frac{C_k}{\beta} (n-1)^{\frac{1}{\beta}-1} = \epsilon_{n-1} +  o(n^{-1-\tilde{\alpha}})
    \end{align*}
    which is summable for $\tilde{\alpha} >0$. Therefore, there always exists an upper bound $\hat{B}$ for $\epsilon_n$ and an upper bound $B$ for $d_n$.
\medskip
    
    \noindent\textbf{Step 3. Upper bound.} Finally, as we showed before from negative to positive the jump will be bounded by $o(n^{-1-\tilde{\alpha}})$. Therefore in order to show $\epsilon_n$ is of $o(n^{-\tilde{\alpha}})$, we only need to do this for the case when $\epsilon_n > 0$ and $d_n>0$, if it's negative then $\epsilon_n$ is less than $\delta_n=o(n^{-\tilde{\alpha}})$. Thus, $\epsilon_n$ can only exceed the $o(n^{-\tilde{\alpha}})$ envelope if it undergoes a contiguous interval where $\epsilon_{n-1} > 0$ and $d_{n-1} > 0$. Then assume for an uninterrupted interval $M>> n > N$ we have $\epsilon_n,d_n > 0$.
    Using the Taylor's expansion upper bound:
    \begin{align*}
         \epsilon_n \leq& \epsilon_{n-1} + \frac{C_k}{2\beta}\left(1-\frac{1}{\beta}\right) (n-1)^{\frac{1}{\beta}-2} -d_{n-1}
         C_k^{\frac{\beta}{\beta-1}}(n-1)^{-1} +\frac{\beta^2}{2(\beta-1)}d_{n-1}^2 C_k^{\frac{\beta+1}{\beta-1}}(n-1)^{-1-\frac{1}{\beta}}\\
         \leq& \epsilon_{n-1} + \frac{C_k}{2\beta}\left(1-\frac{1}{\beta}\right) (n-1)^{\frac{1}{\beta}-2} -d_{n-1}
         C_k^{\frac{\beta}{\beta-1}}(n-1)^{-1} +\frac{\beta^2}{2(\beta-1)}B(|\epsilon_{n-1}|+|\delta_{n-1}|) C_k^{\frac{\beta+1}{\beta-1}}(n-1)^{-1-\frac{1}{\beta}}\\
         =& \epsilon_{n-1}\left(1-\frac{C_k^{\frac{\beta}{\beta-1}}-\frac{\beta^2}{2(\beta-1)}B(n-1)^{-\frac{1}{\beta}} C_k^{\frac{\beta+1}{\beta-1}}}{n-1}\right) + o(n^{-1-\tilde{\alpha}}).
    \end{align*}
    As $C_k^{\frac{\beta}{\beta-1}} \geq 1-\frac{1}{\beta}$, for $n > N$ we will have $C_k^{\frac{\beta}{\beta-1}}-\frac{\beta^2}{2(\beta-1)}B(n-1)^{-\frac{1}{\beta}} C_k^{\frac{\beta+1}{\beta-1}}> \tilde{\alpha}$, then for those $n$ we can write:
    \begin{align*}
         \epsilon_n \leq
         \epsilon_{n-1}\left(1-\frac{D}{n-1}\right) + o(n^{-1-\tilde{\alpha}}).
    \end{align*}
    Again, invoking \Cref{lem:Lower} will conclude the proof.$\Box$
\end{proof}

\end{appendix}
\end{document}